\documentclass[pra,aps,nofootinbib,notitlepage,superscriptaddress,twocolumn]{revtex4-1}

\usepackage{amsthm}
\usepackage{amsmath,bm}
\usepackage{amssymb}
\usepackage{amsfonts}
\usepackage{graphicx}

\usepackage{txfonts}

\usepackage[vcentermath]{youngtab}

\usepackage[colorlinks=true,linkcolor=blue,citecolor=blue,urlcolor=blue]{hyperref}

\newtheorem{lemma}{Lemma}

\DeclareRobustCommand{\nwsearrow}{%
  \mathrel{\text{\ooalign{$\searrow$\cr$\nwarrow$}}}%
}

\begin{document}

\title{Metrological Detection of Multipartite Entanglement from Young Diagrams}

\author{Zhihong Ren}
\affiliation{Institute of Theoretical Physics and Department of Physics, State Key
Laboratory of Quantum Optics and Quantum Optics Devices, Collaborative
Innovation Center of Extreme Optics, Shanxi University, Taiyuan 030006, China}
\affiliation{Laboratoire Kastler Brossel, ENS-Universit\'{e} PSL, CNRS, Sorbonne Universit\'{e}, Coll\`{e}ge de France, 24 Rue Lhomond, 75005, Paris, France}
\author{Weidong Li}
\email{wdli@sxu.edu.cn}
\affiliation{Institute of Theoretical Physics and Department of Physics, State Key
Laboratory of Quantum Optics and Quantum Optics Devices, Collaborative
Innovation Center of Extreme Optics, Shanxi University, Taiyuan 030006, China}
\author{Augusto Smerzi}
\affiliation{Institute of Theoretical Physics and Department of Physics, State Key
Laboratory of Quantum Optics and Quantum Optics Devices, Collaborative
Innovation Center of Extreme Optics, Shanxi University, Taiyuan 030006, China}
\affiliation{QSTAR, INO-CNR, and LENS, Largo Enrico Fermi 2, 50125 Firenze, Italy}
\author{Manuel Gessner}
\email{manuel.gessner@ens.fr}
\affiliation{Laboratoire Kastler Brossel, ENS-Universit\'{e} PSL, CNRS, Sorbonne Universit\'{e}, Coll\`{e}ge de France, 24 Rue Lhomond, 75005, Paris, France}
\date{\today}

\begin{abstract}
We characterize metrologically useful multipartite entanglement by representing partitions with Young diagrams. We derive entanglement witnesses that are sensitive to the shape of Young diagrams and show that Dyson's rank acts as a resource for quantum metrology. Common quantifiers, such as the entanglement depth and $k$-separability are contained in this approach as the diagram's width and height. Our methods are experimentally accessible in a wide range of atomic systems, as we illustrate by analyzing published data on the quantum Fisher information and spin-squeezing coefficients.
\end{abstract}

\maketitle

An efficient classification of entanglement in multipartite systems is crucial for our understanding of quantum many-body systems and the development of quantum information science~\cite{AmicoRMP2006,HaukeNATPHYS2016,FrerotNATCOMMUN2019,GabbrielliNJP2019}. A particular challenge is the development of experimentally implementable criteria for the detection of multipartite entanglement~\cite{GuehneTothPHYSREP2009,FriisNATREV2019}. The development of quantum technologies further demands a precise understanding of the set of multipartite entangled states that enable a quantum advantage in specific applications of quantum information~\cite{ApellanizJPA2014,PezzeRMP2018}. In this context, metrological entanglement criteria are powerful tools that establish a quantitative link between the number of detected entangled parties and the quantum gain in interferometric measurements~\cite{PezzePRL2009,GiovannettiNATPHOTON2011,HyllusTothPRA2012,GessnerPRL2018,PezzeRMP2018}.

As a consequence of the exponentially increasing number of partitions in multipartite systems, there is no unique way to quantify multipartite entanglement. Common approaches to capture the extent of multipartite correlations focus on simple integer indicators~\cite{GuehneTothPHYSREP2009}: An entanglement depth of $w$ describes that at least $w$ parties must be entangled, while $h$-inseparability expresses that the system cannot be split into $h$ separable subsystems. Larger values of $w$ and smaller values of $h$ generally indicate more multipartite entanglement, and experimentally observable bounds on both can be obtained with different methods~\cite{GuehneTothPHYSREP2009,FriisNATREV2019}, including from the metrological sensitivity in terms of the quantum Fisher information~\cite{HyllusTothPRA2012}.

A systematic approach based on the partitions of a multipartite system reveals a duality between $w$ and $h$~\cite{SzylardQUANTUM2019}. Let us illustrate this with the example of a 7-partite system that allows for a separable description in the partition $\Lambda=1|2345|67$, see Fig.~\ref{fig:0}. The system is separable into $h=3$ subsets and it contains entanglement among up to $w=4$ parties, i.e., it has an entanglement depth of $w=4$. By using the correspondence between partitions of a system (up to permutations of the particle labels) and Young diagrams, we can represent this partition as $\Lambda\sim{\tiny{\yng(4,2,1)}}$, where each box represents one party and each row represents an entangled subset of decreasing size from top to bottom. We can easily convince ourselves that $w$ and $h$ correspond to the width and height of the Young diagram, respectively.

\begin{figure}[bt]
\centering
\includegraphics[width=0.35\textwidth]{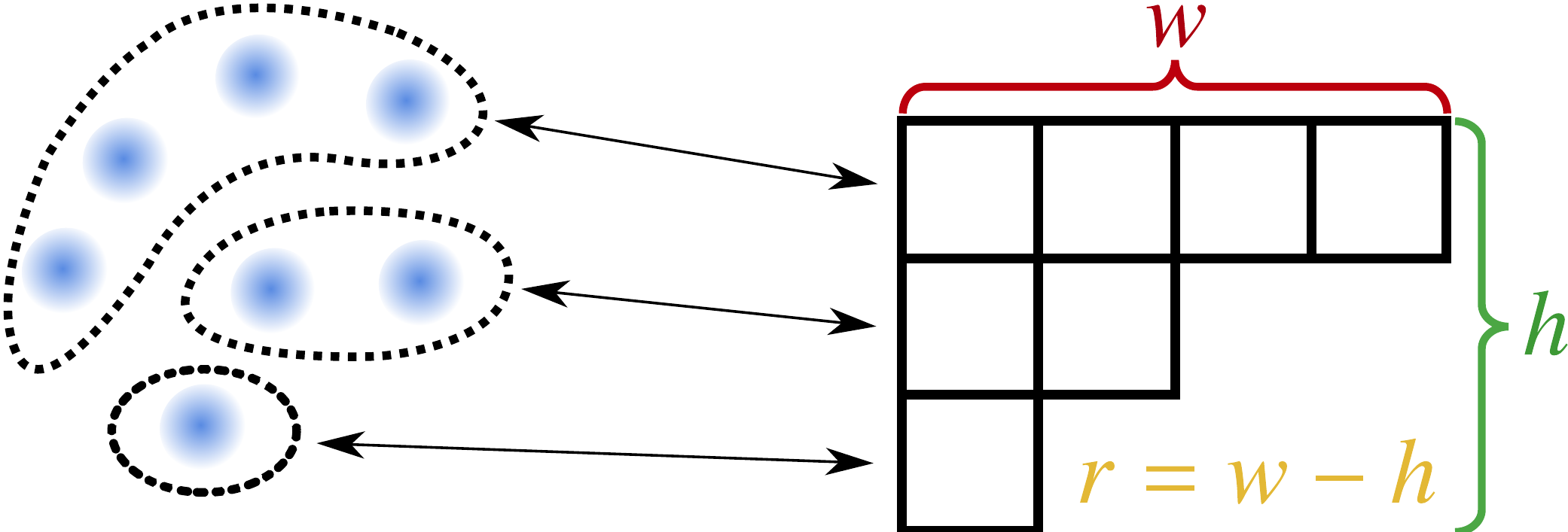}
\caption{Representation of multipartite entanglement using Young diagrams. In this example, a system of $N=7$ particles is separable in a partition into $h=3$ subsets and contains an entanglement depth of $w=4$ particles. These quantities correspond to height $h$ and width $w$ of the associated Young diagram. Dyson's rank $r=w-h$ combines both pieces of information.}
\label{fig:0}
\end{figure}

Focusing exclusively on one of these two quantities provides only limited information about the allowed structure of separable partitions. The entanglement depth $w$ refers to the size of the largest subset but ignores the size and number of the remaining subsets. For instance, $w=4$ does not distinguish between the partition $\Lambda$ and, e.g., $\Lambda'\sim{\tiny{\yng(4,3)}}$, even though the latter clearly contains more entanglement. Separability, in contrast, is insensitive to the size of the entangled subsets and $h=3$ is also compatible with, e.g., $\Lambda''\sim{\tiny{\yng(3,3,1)}}$. As an alternative integer quantifier, the rank of a partition, defined by Dyson~\cite{DysonEUREKA1944} as $r=w-h$, combines the information about $w$ and $h$, and was recently suggested to express the ``strechability'' of correlations~\cite{SzylardQUANTUM2019}. In our example, it successfully distinguishes these partitions and yields the intuitive order $r(\Lambda'')=0$, $r(\Lambda)=1$, and $r(\Lambda')=2$. With $2N-3$ steps between fully separable and genuine $N$-partite entangled states, $r$ provides a scale almost twice as fine as those provided by the $N$ possible values of $w$ and $h$, respectively.

In this work, we derive metrological entanglement criteria that provide combined information about $w$ and $h$, or about Dyson's rank $r$. We base our criteria on  quantifiers of metrological sensitivity that are widely used both in theory and experiments. Our results reveal hidden details about the structure of multipartite entanglement from the quantum Fisher information or spin-squeezing parameters, while relying only on established measurement techniques. This leads to a better understanding of metrologically useful multipartite entanglement, and uncovers in particular the role of $r$ as a resource for quantum-enhanced metrology. As a single integer quantifier, $r$ is found to provide most information about multipartite entanglement in the experimentally relevant regime of limited multipartite entanglement. The entanglement depth $w$, instead, is shown to be most effective close to genuine $N$-partite entanglement.

\textit{Metrological witness for (w,h)-entanglement.}---We characterize the degree of multipartite entanglement in terms of the partitions that are compatible with a separable description of the correlations. A partition $\Lambda=\{A_1,A_2,\dots,A_{|\Lambda|}\}$ separates the total $N$-partite system into $|\Lambda|$ nonempty, disjoint subsets $A_l$ of size $N_l$ such that $\sum_{l=1}^{|\Lambda|}N_l=N$. A state $\hat{\rho}_{\Lambda}$ is $\Lambda$-separable if there exist local quantum states $\hat{\rho}^{(\gamma)}_{A_l}$ for each subsystem and a probability distribution $p_{\gamma}$ such that $\hat{\rho}_{\Lambda}=\sum_{\gamma}p_{\gamma}\hat{\rho}^{(\gamma)}_{A_1}\otimes\cdots\otimes\hat{\rho}^{(\gamma)}_{A_{|\Lambda|}}$. The partitions can be classified according to the number $|\Lambda|$ of subsets and the size $\max \Lambda=\max_l|A_l|$ of the largest subset, i.e., the respective height $h$ and width $w$ of the associated Young diagram. The entanglement depth is defined with respect to the set $\mathcal{L}_{w-\mathrm{prod}}=\{\Lambda\:|\:\max \Lambda\leq w\}$ of partitions with maximal width $w$. Any state that cannot be written as a $w$-producible state $\hat{\rho}_{w-\rm{prod}}=\sum_{\Lambda\in\mathcal{L}_{w-\mathrm{prod}}}P_{\Lambda}\hat{\rho}_{\Lambda}$, where $P_{\Lambda}$ is a probability distribution, has an entanglement depth of at least $w+1$. Analogously, inseparability is related to the set $\mathcal{L}_{h-\mathrm{sep}}=\{\Lambda\:|\:|\Lambda|\geq h\}$ of partitions with minimal height $h$ and $h$-inseparable states cannot be represented in the form $\hat{\rho}_{h-\rm{sep}}=\sum_{\Lambda\in\mathcal{L}_{h-\mathrm{sep}}}P_{\Lambda}\hat{\rho}_{\Lambda}$. By combining both pieces of information, we obtain a finer description of multipartite quantum correlations by the set of $(w,h)$-entangled states, i.e., those that cannot be modeled as $(w,h)$-separable states
\begin{align}
\hat{\rho}_{(w,h)-\rm{sep}}=\sum_{\Lambda\in\mathcal{L}_{w-\mathrm{prod}}\cap\mathcal{L}_{h-\mathrm{sep}}}P_{\Lambda}\hat{\rho}_{\Lambda}.
\end{align}

To derive criteria that allow us to distinguish between different $(w,h)$ classes, we derive the metrological sensitivity limits for states with restricted values of $w$ and $h$. Measuring or calculating the sensitivity of a state then allows us to put bounds on $w$ and $h$ by comparison with these limits. In the following we focus on $N$-qubit systems, described by collective angular momentum operators $\hat{J}_{\mathbf{n}}=\sum_{i=1}^N\mathbf{n}\cdot\hat{\boldsymbol{\sigma}}^{(i)}/2$ with unit vector $\mathbf{n}\in\mathbb{R}^3$ and $\hat{\boldsymbol{\sigma}}^{(i)}=(\hat{\sigma}_x^{(i)},\hat{\sigma}_y^{(i)},\hat{\sigma}_z^{(i)})$ a vector of Pauli matrices for the $i$th qubit. The central theorem of quantum metrology, the quantum Cram\'er-Rao bound $(\Delta\theta_{\mathrm{est}})^2\geq 1/F_Q[\hat{\rho},\hat{J}_{\mathbf{n}}]$, defines the achievable precision limit for the estimation of a phase shift $\theta$ generated by $\hat{J}_{\mathbf{n}}$, using the state $\hat{\rho}$~\cite{BraunsteinPRL1994,GiovannettiPRL2006,GiovannettiNATPHOTON2011,ApellanizJPA2014,PezzeRMP2018}. The phase $\theta$ is estimated from measurements of the quantum state $\hat{\rho}(\theta)=\hat{U}(\theta)\hat{\rho}\hat{U}(\theta)^{\dagger}$ with $\hat{U}(\theta)=e^{-i\hat{J}_{\mathbf{n}}\theta}$ and the quantum Fisher information $F_Q[\hat{\rho},\hat{J}_{\mathbf{n}}]$ describes the sensitivity of $\hat{\rho}(\theta)$ to small variations of $\theta$~\cite{BraunsteinPRL1994}. As an experimentally accessible quantity, $F_Q$ has been employed in the past as a versatile entanglement witness~\cite{PezzePRL2009,HyllusTothPRA2012,PezzeRMP2018,HaukeNATPHYS2016,GabbrielliNJP2019,GessnerPRA2016,StrobelSCIENCE2014,QinNPJQI2019}.

We are now in a position to present the main results of this work. The quantum Fisher information of any $(w,h)$-separable state is limited to
\begin{align}\label{eq:Fwh}
F_Q[\hat{\rho}_{(w,h)-\rm{sep}},\hat{J}_{\mathbf{n}}]\leq w(N-h)+N,
\end{align}
where $N/h\leq w \leq N-h+1$ and $N/w\leq h \leq N-w+1$. The bound~(\ref{eq:Fwh}) can be slightly tightened by explicitly considering the division of $N$ into integer subsets, and in this case it is saturated by an optimal quantum state. A detailed proof of~(\ref{eq:Fwh}) in its most general form, as well as the optimal states are provided in~\cite{Supp}. The monotonic growth of Eq.~(\ref{eq:Fwh}) in $w$ and its monotonic decrease in $h$ demonstrate that higher quantum advantages in metrology measurements require entanglement among larger sets of particles.

In the extreme cases where all or none of the parties are entangled, we recover the well-known limits of classical and quantum parameter estimation strategies, respectively~\cite{GiovannettiPRL2006}. Fully separable states, defined by $(w,h)=(1,N)$, are limited to $F_Q[\hat{\rho}_{(1,N)-\rm{sep}},\hat{J}_{\mathbf{n}}]\leq F_{\rm SN}[\hat{J}_{\mathbf{n}}]=N$, which leads to shot-noise sensitivity $(\Delta \theta_{\rm{est}})^2\geq 1/N$, while genuine $N$-partite entanglement, $(w,h)=(N,1)$, enables sensitivities up to the Heisenberg limit $F_Q[\hat{\rho}_{(N,1)-\rm{sep}},\hat{J}_{\mathbf{n}}]\leq N^2$ with $(\Delta \theta_{\rm{est}})^2\geq 1/N^2$~\cite{GiovannettiNATPHOTON2011,PezzePRL2009,PezzeRMP2018}. In between these extreme cases, the metrological potential of finitely entangled states is captured by the combined information provided by the tuple $(w,h)$. The metrological entanglement witness~(\ref{eq:Fwh}) has a particularly simple interpretation: It identifies the quantum advantage offered by $(w,h)$-entanglement in terms of the sensitivity difference to the shot-noise limit, $Q[\hat{\rho},\hat{J}_{\mathbf{n}}]:=F_Q[\hat{\rho},\hat{J}_{\mathbf{n}}]-F_{\rm SN}[\hat{J}_{\mathbf{n}}]$. The advantage is indeed bounded for $(w,h)$-separable states by
\begin{align}\label{eq:Qwh}
Q[\hat{\rho}_{(w,h)-\rm{sep}},\hat{J}_{\mathbf{n}}]\leq w(N-h).
\end{align}
A sensitivity that exceeds the shot-noise limit beyond this bound consequently implies metrologically useful $(w,h)$-entanglement.

\begin{figure}[tb]
\centering
\includegraphics[width=0.47\textwidth]{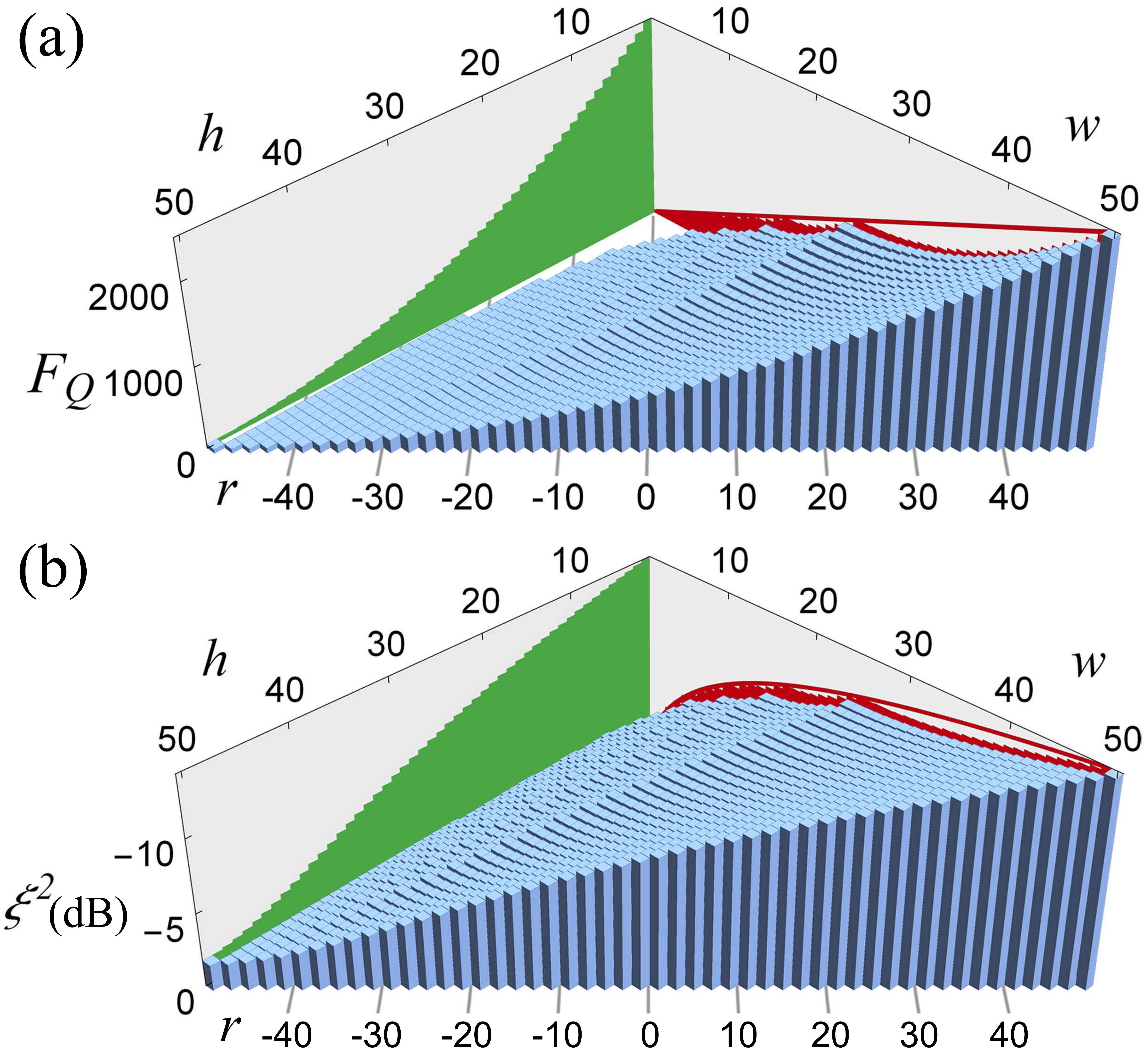}
\caption{Maximal metrological sensitivity for $(w,h)$-separable states as a function of $w$ and $h$. The sensitivity limits are given for (a) the quantum Fisher information $F_Q[\hat{\rho},\hat{J}_{\mathbf{n}}]$ and (b) the spin-squeezing parameter $\xi^2$. By ignoring either $w$ or $h$, we obtain bounds for $w$-producible (red projection) and $h$-separable states (green projection), respectively. A finer projection is given by Dyson's rank $r=w-h$~\cite{footnote}.}
\label{fig:1}
\end{figure}

We recover known bounds that only provide information on either $w$ or $h$ by ignoring part of the information contained in~(\ref{eq:Fwh}). For instance, by replacing $h$ with the trivial lower bound $N/w$, we obtain the well-known sensitivity limit of $w$-producible states~\cite{HyllusTothPRA2012}
\begin{align}\label{eq:Fw}
F_Q[\hat{\rho}_{w-\rm{prod}},\hat{J}_{\mathbf{n}}]\leq wN,
\end{align}
where $1\leq w\leq N$. The result~(\ref{eq:Fwh}) thus generalizes~(\ref{eq:Fw}) which has enabled the widespread study of multiparticle entanglement in quantum metrology~\cite{PezzeRMP2018}, but also provides a valuable tool to understand entanglement in quantum-many body systems~\cite{HaukeNATPHYS2016,GabbrielliNJP2019} and topological quantum phase transitions~\cite{PezzePRL2017}. Similarly, we can ignore the information about $w$ by using the trivial upper bound $N-h+1$, yielding the sensitivity limit of $h$-separable states~\cite{HongPRA2015}
\begin{align}\label{eq:Fh}
F_Q[\hat{\rho}_{h-\rm{sep}},\hat{J}_{\mathbf{n}}]\leq (N-h+1)^2+h-1,
\end{align}
where $1\leq h\leq N$.

Rather than fully ignoring the information provided by either $w$ or $h$, we combine both into a more informative integer quantifier of multipartite entanglement. Dyson's rank $r=w-h$ reflects the increase of correlations due to both larger $w$ and smaller $h$. The range of $r$ are the integer values from $-(N-1)$ to $N-1$ except $\pm(N-2)$~\cite{DysonEUREKA1944}. The set of states with Dyson's rank not larger than $r$ is defined as $\hat{\rho}_{r-\rm{rnk}}=\sum_{\Lambda\in\mathcal{L}_{r-\mathrm{rnk}}}P_{\Lambda}\hat{\rho}_{\Lambda}$ via the set $\mathcal{L}_{r-\mathrm{rnk}}=\{\Lambda\:|\:\max \Lambda-|\Lambda|\leq r\}$~\cite{SzylardQUANTUM2019}. We obtain the bound~\cite{Supp}
\begin{align}\label{eq:Fr}
F_Q[\hat{\rho}_{r-\rm{rnk}},\hat{J}_{\mathbf{n}}]\leq \frac{(N+r)^2}{4}-\frac{1}{4}+N,
\end{align}
for all values of $r$, except for $r=4-N$, where we have $F_Q[\hat{\rho}_{(4-N)-\rm{rnk}},\hat{J}_{\mathbf{n}}]\leq N+4$. The first term in~(\ref{eq:Fr}) clearly identifies the quadratic quantum advantage over the shot-noise limit offered by states with larger Dyson's rank $r$ in terms of $Q[\hat{\rho},\hat{J}_{\mathbf{n}}]$.

The upper bounds for $(w,h)$-separable states given in~(\ref{eq:Fwh}) are represented as a function of $w$ and $h$ in Fig.~\ref{fig:1}~(a) \cite{footnote}. The bounds on producibility~(\ref{eq:Fw}) and separability~(\ref{eq:Fh}) are recovered as the projections onto the axes describing $w$ or $h$ (red and green plots, respectively). Since these correspond to the short arms of the right triangle (blue columns) that is occupied by tuples $(w,h)$, these projections ignore large amounts of information on the respective other coordinate. A finer resolution can be obtained by the projection along the hypothenuse that is described by $r=w-h$. The most detailed information about multipartite entanglement is provided by the tuples $(w,h)$.

\begin{figure}[tb]
\centering
\includegraphics[width=0.49\textwidth]{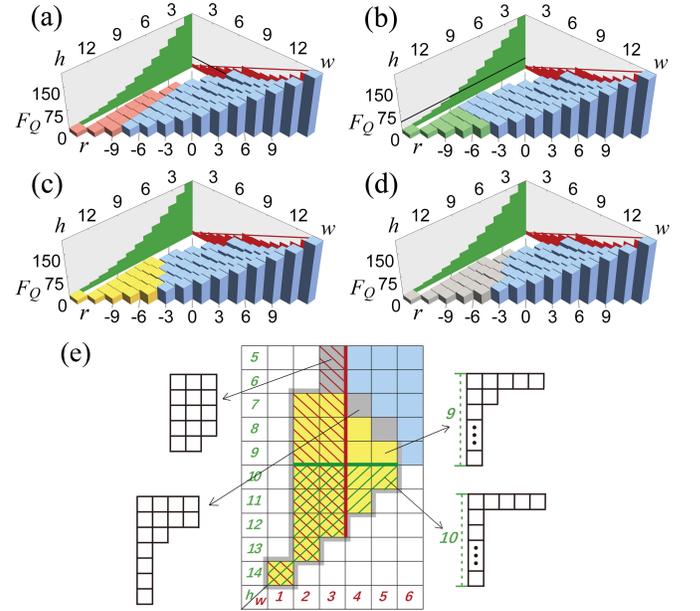}
\caption{Tuples $(w,h)$ whose separability can be excluded from the experimentally extracted value of $F_Q\geq 40.4$ for $N=14$ parties in Ref.~\cite{MonzPRL2011} are highlighted. The red tuples (a) are excluded from the entanglement depth~(\ref{eq:Fw}), the green tuples (b) from $h$-separability~(\ref{eq:Fh}) and the yellow tuples (c) from Dyson's rank~(\ref{eq:Fr}). The grey tuples~(d) are found inseparable by~(\ref{eq:Fwh}), that uses the full information on $(w,h)$~\cite{footnote}. The results are summarized in the top-view plot (e), where each square represents one tuple $(w,h)$ and example partitions with given width $w$ and height $h$ are illustrated.}
\label{fig:2}
\end{figure}

\textit{Analysis of experimental data.---}Our results allow us to extract information about these quantities directly from $F_Q$ without the need for additional measurements. To illustrate the power of this technique, we study the separability structure of experimentally generated quantum states based on published lower bounds for $F_Q$. We first focus on an example with moderate particle number $N=14$, reported in~\cite{MonzPRL2011}. In this trapped-ion experiment a quantum Fisher information of at least $F_Q[\hat{\rho},\hat{J}_z]\geq 40.4$ has been observed~\cite{PezzePNAS2016}. The performance of the different bounds can be gauged by the number of separable $(w,h)$ classes, i.e., tuples $(w,h)$ that are excluded. Note that more than one partition may be compatible with a tuple $(w,h)$. From Eq.~(\ref{eq:Fw}) and its sharper version~\cite{Supp}, we find that the measured data is incompatible with partitions of width $w\leq 3$, implying an entanglement depth of $w=4$, which excludes $16$ tuples (Fig.~\ref{fig:2}a). Similarly, from~(\ref{eq:Fh}) we find $h=9$, excluding the $11$ tuples of the system into more than $10$ parts (Fig.~\ref{fig:2}b). Much more information is obtained by using the bound~(\ref{eq:Fwh}), which excludes a total of $24$ separable tuples (Fig.~\ref{fig:2}d)~\cite{footnote}. Among all single integer quantifiers, Dyson's rank, obtained from~(\ref{eq:Fr}) to be $r=-3$, detects the largest amount of $20$ separable tuples (Fig.~\ref{fig:2}c)~\cite{footnote}. The excluded tuples for each criterion are summarized in Fig.~\ref{fig:2}~(e), where we also highlight specific inseparable partitions that remain undetected by the individual information on $w$ or $h$.

\begin{figure}[tb]
\centering
\includegraphics[width=0.49\textwidth]{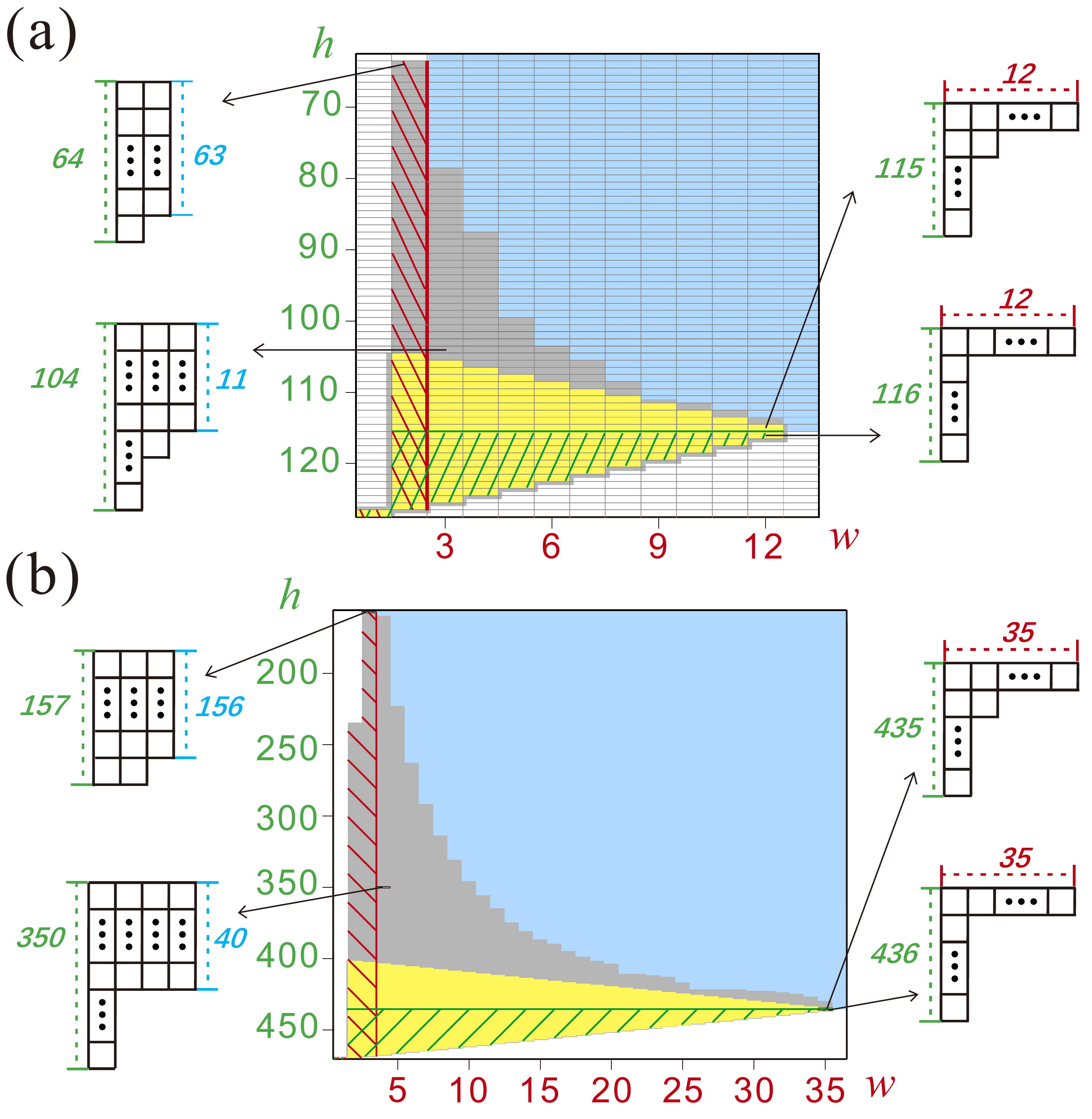}
\caption{(a) Analysis of the multipartite entanglement structure (same as in Fig.~\ref{fig:2}e) for $N=127$ parties based on the measurement of $F_Q[\hat{\rho},\hat{J}_z]\geq 266.7$ reported in Ref.~\cite{BohnetSCIENCE2016}. (b) Same analysis based on the spin-squeezing coefficient~(\ref{eq:xiFQmax}) for a measurement of $\xi^2\leq-4.5\,\mathrm{dB}$ with $N=470$ particles~\cite{StrobelSCIENCE2014}.}
\label{fig:9}
\end{figure}

This technique may also be applied to systems with larger particle numbers, as we illustrate through the analysis of additional data on measurements of lower bounds on the quantum Fisher information in systems of cold atoms and ions published in Refs.~\cite{MonzPRL2011,BarontiniSCIENCE2015,BohnetSCIENCE2016}. While for a full account, we refer to~\cite{Supp}, in Fig.~\ref{fig:9}~(a) we show results obtained from the measurement of $F_Q[\hat{\rho},\hat{J}_z]\geq 266.7$ with $N=127$ trapped ions, as announced in~\cite{BohnetSCIENCE2016}. Generally, we find that $r$ is the most sensitive single integer quantifier of multipartite entanglement in the experimentally most relevant regime of finite entanglement. Only close to the limit of genuine multipartite entanglement, the entanglement depth $w$ becomes slightly more sensitive than $r$, while at no point do we gain more information from $h$~\cite{Supp}. 

\textit{Bounds for the spin-squeezing coefficient.}---We have so far focused on metrological entanglement witnesses that make use of the quantum Fisher information, the ultimate sensitivity limit achievable by an optimal measurement. In many experimental situations, it is more convenient to study the precision with respect to the specific measurement of a collective spin observable. This is achieved by spin-squeezing coefficients~\cite{WinelandPRA1994,MaPHYSREP2011,PezzeRMP2018}, first introduced by Wineland \textit{et al.} as $\xi^2=N(\Delta \hat{J}_{\mathbf{n}})^2/\langle \hat{J}_{\mathbf{m}}\rangle^2$, with suitably chosen, orthogonal directions $\mathbf{n}$ and $\mathbf{m}$~\cite{WinelandPRA1994}. The spin-squeezing coefficient expresses the quantum gain in sensitivity over the shot-noise limit due to squeezing of a spin observable $\hat{J}_{\mathbf{n}}$ and has found widespread application in experiments with atomic systems~\cite{PezzeRMP2018}. Spin squeezing further gives rise to lower bounds on the quantum Fisher information~\cite{PezzePRL2009,GessnerPRL2019} and provides an experimentally convenient witness for the entanglement depth $w$~\cite{SorensenNATURE2001,SorensenPRL2001,ApellanizJPA2014,FadelPRA2020}. We derive state-independent bounds on $\xi^2$ that are sensitive to both $w$ and $h$, by relating the spin-squeezing coefficient to the bounds that we found for the quantum Fisher information. Specifically, we show that~\cite{Supp}
\begin{align}\label{eq:xiFQmax}
N\max_{\hat{\rho}_{\Lambda-\rm{sep}}}\xi_{\hat{\rho}_{\Lambda-\rm{sep}}}^{-2}\leq \frac{1}{2} \max_{\hat{\rho}_{\Lambda-\rm{sep}}}F_Q[\hat{\rho}_{\Lambda-\rm{sep}},\hat{J}_{\mathbf{n}}]+N,
\end{align}
which allows us to use our results~(\ref{eq:Fwh}), (\ref{eq:Fw}), (\ref{eq:Fh}), and~(\ref{eq:Fr}) on $F_Q$ to identify limits on $\xi^2$ as a function of $w$, $h$ or $r$~\cite{Supp}. These bounds are shown in Fig.~\ref{fig:1}~(b). For instance, from~(\ref{eq:Fr}) we obtain the limit
\begin{align}\label{eq:xir}
\xi_{\rho_{r-\mathrm{rnk}}}^{2}&\geq \frac{8N}{(N+r)^2+12N-1}
\end{align}
for states with Dyson's rank no larger than $r$. In Fig.~\ref{fig:9}~(b) we summarize the entanglement analysis based on the experimentally measured value of $\xi^2\leq-4.5\,\mathrm{dB}$ of spin squeezing for $N=470$ particles, reported in~\cite{StrobelSCIENCE2014}.

\begin{figure}[tb]
\centering
\includegraphics[width=0.49\textwidth]{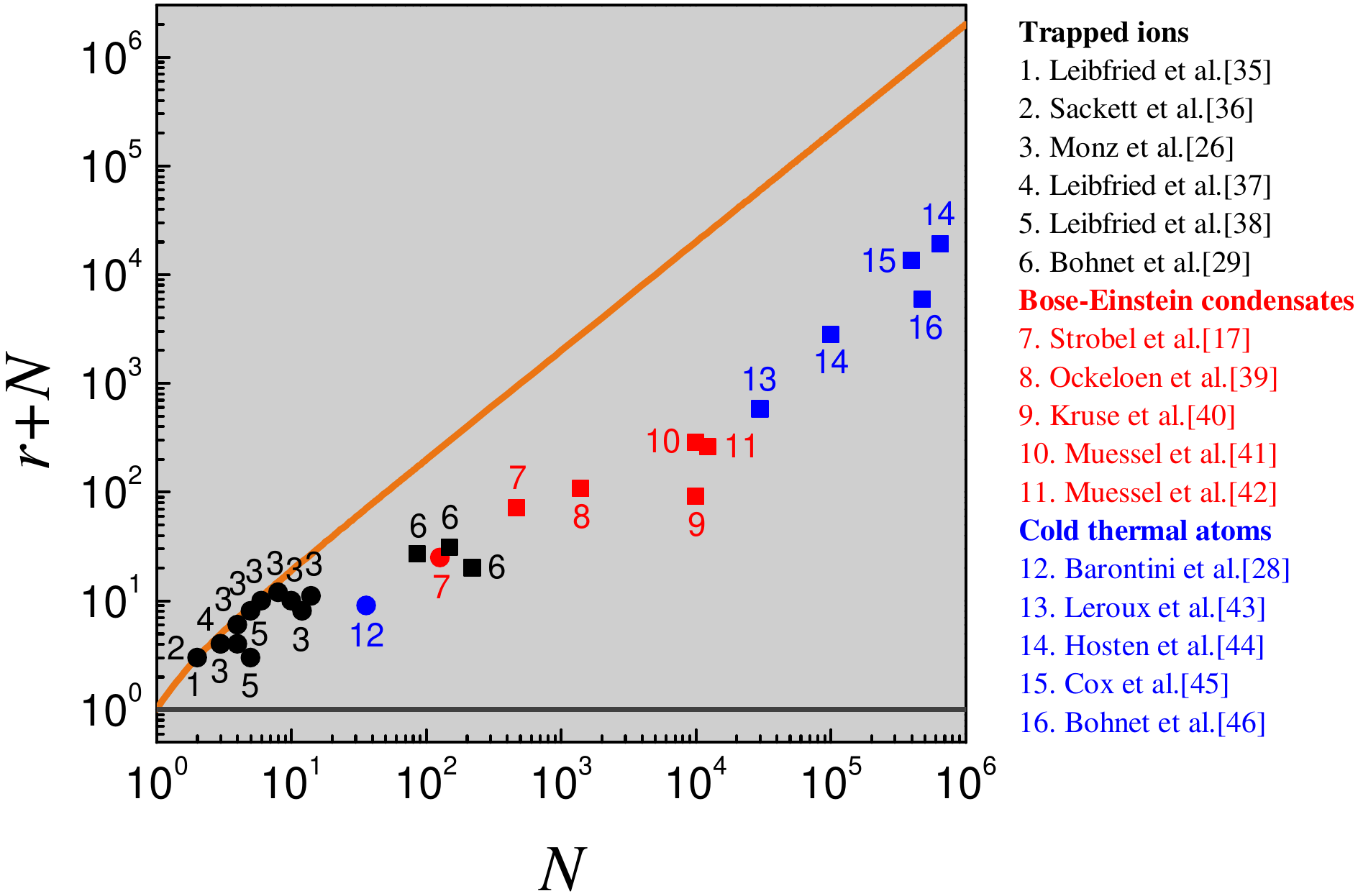}
\caption{Dyson's rank $r$ as a function of $N$ extracted from published experimental data on the spin-squeezing coefficient $\xi^2$ (squares) or the quantum Fisher information $F_Q$ (circles) based the results reported in Refs.~\cite{MonzPRL2011,BarontiniSCIENCE2015,BohnetSCIENCE2016,StrobelSCIENCE2014,SackettNature2000,LeibfriedNature2003,LeibfriedScience2004,LeibfriedNature2005,OckeloenPRL2013,MuesselPRL2014,MuesselPRA2015,KrusePRL2016,LerouxPRL2010a,BohnetNPHOTON2014,HostenNature2016,CoxPRL2016,PezzeRMP2018,PezzePNAS2016}. The quantity $r+N$ ranges between $1$ for fully separable states (black line) and $2N-1$ for genuine $N$-partite entangled states (orange line).}
\label{fig:rsummary}
\end{figure}

From the widely available experimental data on $\xi^2$ and $F_Q$, we can immediately extract measured values of Dyson's rank $r$ using the bounds~(\ref{eq:Fr}) and~(\ref{eq:xir})~\cite{footnote}. In Fig.~\ref{fig:rsummary} we analyze experimental data of Refs.~\cite{MonzPRL2011,BarontiniSCIENCE2015,BohnetSCIENCE2016,StrobelSCIENCE2014,SackettNature2000,LeibfriedNature2003,LeibfriedScience2004,LeibfriedNature2005,OckeloenPRL2013,MuesselPRL2014,MuesselPRA2015,KrusePRL2016,LerouxPRL2010a,BohnetNPHOTON2014,HostenNature2016,CoxPRL2016,PezzeRMP2018,PezzePNAS2016}, describing experiments with trapped ions, Bose-Einstein condensates and cold thermal atoms. In order to be able to compare the measurements with different numbers of particles $N$, we plot $r+N$ with values in the range $[1,2N-1]$.

\textit{Conclusions.}---Based on widely used quantifiers for metrological sensitivity, we have derived entanglement witnesses that are sensitive to combined constraints on the size $w$ of the largest entangled group ($w$-producibility or entanglement depth) and the number $h$ of separable groups ($h$-separability). The description of inseparable partitions in terms of Young diagrams has allowed us to gain a precise understanding of metrologically useful multipartite entanglement beyond the information that can be provided by either $w$ or $h$ individually. Our techniques can be readily implemented for the theoretical and experimental study of multipartite entanglement in quantum information and many-body physics.

\textit{Acknowledgments.}---This research was supported by the National Natural Science Foundation of China (Grant No. 11874247), National Key R\&D Program of China (Grant No. 2017YFA0304500), 111 project (Grant No. D18001), the Hundred Talent Program of the Shanxi Province (2018), and by the LabEx ENS-ICFP: ANR-10-LABX-0010/ANR-10-IDEX-0001-02 PSL*.

\onecolumngrid

\clearpage 
\twocolumngrid

\renewcommand\thefigure{S\arabic{figure}} 

\appendix
\begin{center}\large{\textbf{Supplemental Material}}\end{center}
\section{Sensitivity limits on $(w,h)$-separable states}
We derive the sensitivity limit for any $(w,h)$-separable state $\hat{\rho}_{(w,h)-\rm{sep}}=\sum_{\Lambda\in\mathcal{L}_{w-\mathrm{prod}}\cap\mathcal{L}_{h-\mathrm{sep}}}P_{\Lambda}\hat{\rho}_{\Lambda}$, i.e., Eq.~(\ref{eq:Fwh}) in the main text and its more general version. To this end, we first derive the sensitivity limit for states that are separable in an arbitrary partition $\Lambda$. The convexity of the quantum Fisher information~\cite{Varenna} implies that
\begin{align}\label{eq:FLambdawh}
F_Q[\hat{\rho}_{(w,h)-\rm{sep}},\hat{J}_{\mathbf{n}}]\leq \sum_{\Lambda\in\mathcal{L}_{w-\mathrm{prod}}\cap\mathcal{L}_{h-\mathrm{sep}}}P_{\Lambda}F_Q[\hat{\rho}_{\Lambda},\hat{J}_{\mathbf{n}}].
\end{align}
Moreover, let $\Lambda\in\mathcal{L}_{w-\mathrm{prod}}\cap\mathcal{L}_{h-\mathrm{sep}}$ and $\hat{\rho}_{\Lambda}=\sum_{\gamma}p_{\gamma}\hat{\rho}^{(\gamma)}_{A_1}\otimes\cdots\otimes\hat{\rho}^{(\gamma)}_{A_{|\Lambda|}}$. We obtain, again, from convexity that
\begin{align}\label{eq:Fconvex}
F_Q[\hat{\rho}_{\Lambda},\hat{J}_{\mathbf{n}}]\leq \sum_{\gamma}p_{\gamma}F_Q[\hat{\rho}^{(\gamma)}_{A_1}\otimes\cdots\otimes\hat{\rho}^{(\gamma)}_{A_{|\Lambda|}},\hat{J}_{\mathbf{n}}].
\end{align}
We may decompose the collective angular momentum operator into the subsets defined by $\Lambda$ as
\begin{align}
\hat{J}_{\mathbf{n}}=\sum_{l=1}^{|\Lambda|}\hat{J}_{A_l,\mathbf{n}},
\end{align}
where
\begin{align}
\hat{J}_{A_l,\mathbf{n}}=\frac{1}{2}\sum_{i\in A_l}\mathbf{n}\cdot\hat{\boldsymbol{\sigma}}^{(i)}
\end{align}
is an operator of total angular momentum $N_l/2$ and $N_l=\sum_{i\in A_l}$ is the number of qubits in the subset $A_l$. Additivity~\cite{Varenna} now implies that
\begin{align}\label{eq:Fadditive}
F_Q[\hat{\rho}^{(\gamma)}_{A_1}\otimes\cdots\otimes\hat{\rho}^{(\gamma)}_{A_{|\Lambda|}},\hat{J}_{\mathbf{n}}]=\sum_{l=1}^{|\Lambda|}F_Q[\hat{\rho}^{(\gamma)}_{A_l},\hat{J}_{A_l,\mathbf{n}}].
\end{align}
The quantum Fisher information obeys the following sequence of bounds~\cite{Varenna},
\begin{align}\label{eq:QFIVarbound}
F_Q[\hat{\rho}^{(\gamma)}_{A_l},\hat{J}_{A_l,\mathbf{n}}]&\leq 4(\Delta \hat{J}_{A_l,\mathbf{n}})^2_{\hat{\rho}^{(\gamma)}_{A_l}}\notag\\
&\leq N_l^2.
\end{align}
The first bound is saturated by all pure states while the second is achieved if and only if $\hat{\rho}^{(\gamma)}_{A_l}$ is an $N_l$-qubit Greenberger-Horne-Zeilinger (GHZ) state $|\mathrm{GHZ}_{N_l}\rangle=({|\!\downarrow\rangle^{\otimes N_l}}+e^{i\phi}{|\!\uparrow\rangle^{\otimes N_l}})/\sqrt{2}$ with arbitrary phase $\phi$. Inserting~(\ref{eq:QFIVarbound}) and~(\ref{eq:Fadditive}) into~(\ref{eq:Fconvex}) yields the upper bound
\begin{align}\label{eq:maxLambda}
F_Q[\hat{\rho}_{\Lambda-\rm{sep}},\hat{J}_{\mathbf{n}}]\leq \sum_{l=1}^{|\Lambda|}N_l^2.
\end{align}
This tight bound can thus be used to check the inseparability of each Young diagram individually: Each row with width $N_l$ of the diagram can contribute a sensitivity of at most $N_l^2$.

\begin{figure}[tb]
\centering
\includegraphics[width=0.25\textwidth]{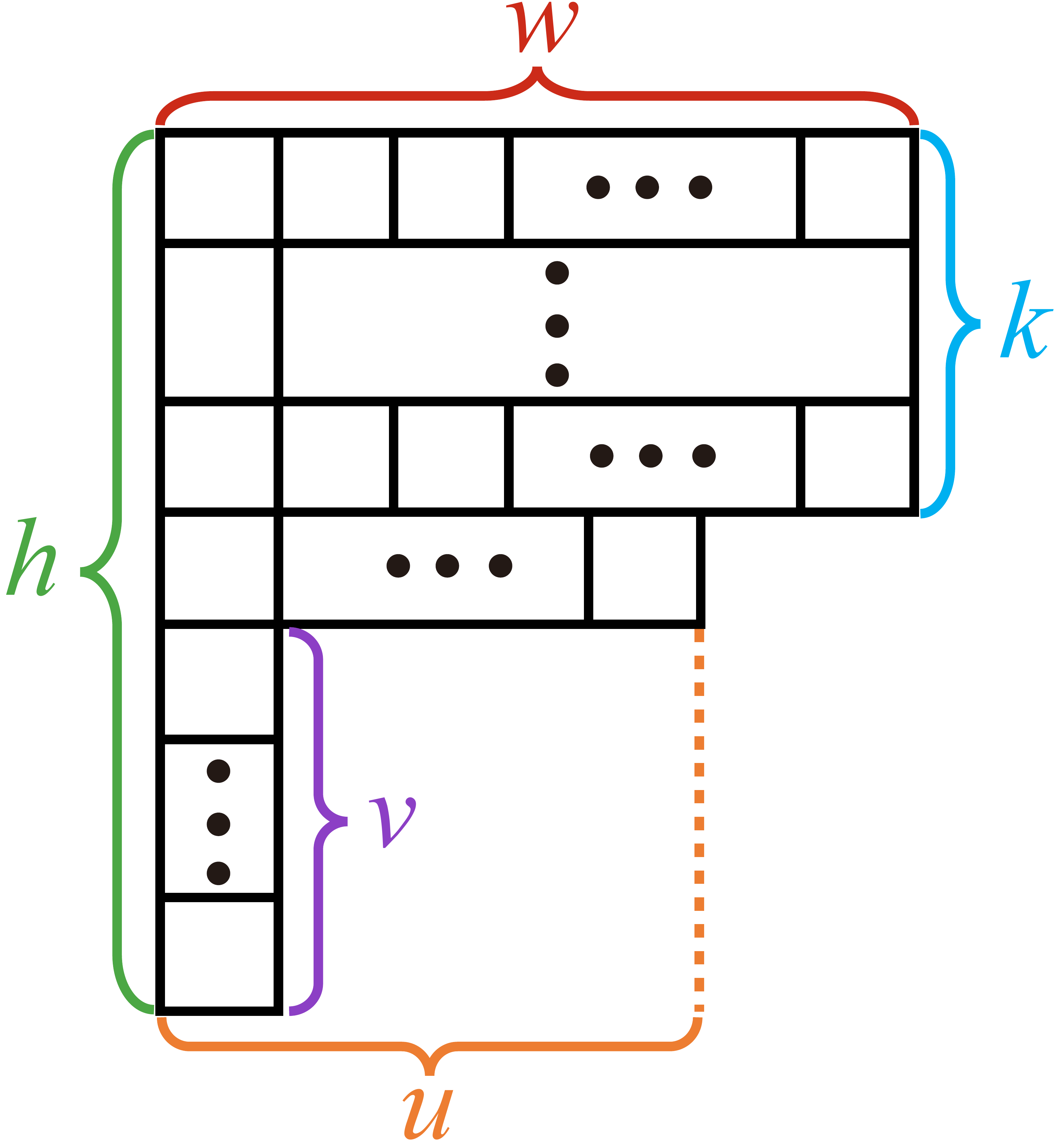}
\caption{
Construction of the partition that maximizes $f_N(w,h)$, defined in Eq.~(%
\protect\ref{eq:fNwh}). See the proof of Lemma~\protect\ref{lm:fwh} for
details.}
\label{fig:3}
\end{figure}

To optimize over entire classes of diagrams, we now introduce the functions
\begin{align}\label{eq:fNwh}
f_N(w,h):=\max_{\Lambda\in\mathcal{L}_{w-\mathrm{prod}}\cap\mathcal{L}_{h-\mathrm{sep}}} \sum_{l=1}^{|\Lambda|}N_l^2.
\end{align}
By construction, we have that $F_Q[\hat{\rho}_{\Lambda},\hat{J}_{\mathbf{n}}]\leq f_N(w,h)$ for all $\Lambda\in\mathcal{L}_{w-\mathrm{prod}}\cap\mathcal{L}_{h-\mathrm{sep}}$, and we thus obtain from Eq.~(\ref{eq:FLambdawh})
\begin{align}\label{eq:FQfNwh}
F_Q[\hat{\rho}_{(w,h)-\rm{sep}},\hat{J}_{\mathbf{n}}]\leq f_N(w,h).
\end{align}
This bound can be saturated by a pure product of GHZ-states $|\Psi_N(w,h)\rangle=\bigotimes_{l=1}^{|\Lambda|}|\mathrm{GHZ}_{N_l}\rangle$, where $\Lambda$ here is the partition that achieves the maximum in~(\ref{eq:fNwh}). We can therefore write equivalently
\begin{align}\label{eq:FQmaxwh}
\max_{\hat{\rho}_{(w,h)-\rm{sep}}}F_Q[\hat{\rho}_{(w,h)-\rm{sep}},\hat{J}_{\mathbf{n}}]=f_N(w,h),
\end{align}
where the maximization includes all $(w,h)$-separable states. The sensitivity limits of $(w,h)$-separable states are therefore determined by the functions $f_N(w,h)$.\begin{lemma}\label{lm:fwh}
The function $f_N(w,h)$ evaluates to
\begin{subequations}\label{eq:fNwhfull}
\begin{align}
f_N(w,h)=kw^{2}+u^{2}+v,
\end{align}%
with
\begin{align}
k&=\left\lfloor \frac{N-h}{w-1}\right\rfloor,\\
u&=N-h+1-(w-1)k,\\
v&=h-k-1.
\end{align}
\end{subequations}
\end{lemma}
\begin{proof}
To identify the maximally sensitive partition, it is instructive to picture partitions in terms of their associated Young diagrams. As $(N_{1}+1)^{2}+(N_{2}-1)^{2}> N_{1}^{2}+N_{2}^{2}$ if $N_{1}>N_{2}$, the sum $s(\Lambda)=\sum_{l=1}^{|\Lambda|}N_l^2$ is maximized by choosing a partition $\Lambda$ with as many parties as possible in the top rows of its associated Young diagram.

Let us first focus on partitions with fixed $(w,h)$. The optimal partition, which achieves the maximum in $s(\Lambda)$, thus has $k\geq 1$ filled rows of maximal width $w$, one partially
filled row with $1\leq u<w$ parties, and the remaining $v$ rows consist of single parties only, see Fig.~\ref{fig:3}. To determine the value of $k$, recall that the fixed height takes up $h$ particles for the left-most column of the Young diagram. The
remaining $N-h$ parties are distributed into a maximum number of $k=\left\lfloor \frac{N-h}{w-1}\right\rfloor $ rows with width $w-1$, where $\lfloor \cdot\rfloor$ indicates the floor function. A number of $u-1=N-h-k(w-1)$ remains when $w-1$ is not a divisor of $N-h$. The total number of particles is
given by $N=kw+u+v$ which leads to $v=h-k-1$.

The optimization in~(\ref{eq:fNwh}) also involves partitions with smaller $w$ and larger $h$. However, these cannot exceed the bound since it increases with $w$ and decreases with $h$ and therefore applies to all partitions in the set $\Lambda_{w-\mathrm{prod}}\cap\Lambda_{h-\mathrm{sep}}$.
\end{proof}
From the construction of the proof it further becomes evident that $f_N(w,h)$ is strictly monotonically increasing with $w$ and strictly monotonically decreasing with $h$.

Inserting Eqs.~(\ref{eq:fNwhfull}) into~(\ref{eq:FQmaxwh}) yields the most general expression for the maximal sensitivity of $(w,h)$-separable states. This result is shown by the blue columns in Figs.~\ref{fig:1} (a), Fig.~\ref{fig:2} (a)-(d) in the main text, and in the examples analyzed below. A simpler expression is obtained by ignoring the separation into integer sets:
\begin{lemma} The function $f_N(w,h)$ satisfies the upper bound
\begin{align}\label{eq:fNwhupper}
f_N(w,h)\leq w(N-h)+N,
\end{align}
\end{lemma}
This implies the result $F_Q[\hat{\rho}_{(w,h)-\rm{sep}},\hat{J}_{\mathbf{n}}]\leq w(N-h)+N$ that was given in Eq.~(\ref{eq:Fwh}) in the main text.
\begin{proof}
We rewrite~(\ref{eq:fNwhfull}) as $f_{N}(w,h)=(w-1)^{2}k^{2}+bk+c$ with $b=(w-1)(w-1+2h-2N)$, and $c=(N-h)^{2}+2N-h$. As a function of $k$ this parabola is strictly monotonically increasing for $k>k_0-1/2$ with $k_0=\frac{N-h}{w-1}$. Since $k_0-1<k\leq k_0$, we obtain an upper bound by replacing $k$ by $k_0$, yielding the expression~(\ref{eq:fNwhupper}).
\end{proof}

It is interesting to maximize $f_N(w,h)$ over either one of the two arguments to identify the sensitivity limits of $w$-producible or $h$-separable states. The maximum values are obtained when $h$ is smallest or $w$ is largest. The smallest value is given by $h=\left\lceil \frac{N}{w}\right\rceil$, where $\lceil\cdot\rceil$ is the ceiling function. Inserting this into Eq.~(\ref{eq:fNwhfull}) yields
\begin{align}\label{eq:flrw}
f^{\leftrightarrow}_N(w)=\max_hf_N(w,h)=sw^{2}+t^{2},
\end{align}
where $s=\left\lfloor \frac{N}{w}\right\rfloor $ and $t=N-\left\lfloor \frac{%
N}{w}\right\rfloor w$. As we may equally write $f^{\leftrightarrow}_N(w)=\max_{\Lambda\in\mathcal{L}_{w-\mathrm{prod}}}\sum_{l=1}^{|\Lambda|}N_l^2$, this implies the well-known result~\cite{HyllusTothPRA2012}
\begin{align}\label{eq:FQmaxw}
\max_{\hat{\rho}_{w-\rm{prod}}}F_Q[\hat{\rho}_{w-\rm{prod}},\hat{J}_{\mathbf{n}}]= f^{\leftrightarrow}_N(w),
\end{align}
which is seen as the red projection in the back of Fig.~2 (a) in the main manuscript. By ignoring the fact that for non-integer $N/w$ we cannot divide the set of $N$ particles into $w$ groups, i.e., by removing the floor function in the above expression, we obtain the simpler upper bound $\max_{\hat{\rho}_{w-\rm{prod}}}F_Q[\hat{\rho}_{w-\rm{prod}},\hat{J}_{\mathbf{n}}]\leq wN$, as stated in Eq.~(4) in the main manuscript.

Similarly, by replacing in Eq.~(\ref{eq:fNwhfull}) $w$ by its largest possible value $w=N+1-h$, we obtain the bound
\begin{align}\label{eq:fudh}
f^{\updownarrow}_N(h)=\max_wf_N(w,h)=(N+1-h)^{2}+h-1.
\end{align}
With $f^{\updownarrow}_N(h)=\max_{\Lambda\in\mathcal{L}_{h-\mathrm{sep}}}\sum_{l=1}^{|\Lambda|}N_l^2$, this implies~\cite{HongPRA2015}
\begin{align}\label{eq:FQmaxh}
\max_{\hat{\rho}_{h-\rm{sep}}}F_Q[\hat{\rho}_{h-\rm{sep}},\hat{J}_{\mathbf{n}}]=f^{\updownarrow}_N(h),
\end{align}
which is shown as green projection in Fig.~2~(a) in the main manuscript.

\section{Sensitivity limits on states with Dyson's rank $r$}
To determine the sensitivity limit on states with bounded Dyson's rank, we proceed analogously as before and we introduce the functions
\begin{align}\label{eq:fNr}
f^{\nwsearrow}_N(r):=\max_{\Lambda\in\mathcal{L}_{r-\mathrm{rnk}}}\sum_{l=1}^{|\Lambda|}N_l^2,
\end{align}
leading to
\begin{align}\label{eq:FQmaxr}
\max_{\hat{\rho}_{r-\rm{rnk}}}F_Q[\hat{\rho}_{r-\rm{rnk}},\hat{J}_{\mathbf{n}}]=f^{\nwsearrow}_N(r).
\end{align}
\begin{lemma}\label{lm:fr}
For $N+r$ odd, the function $f^{\nwsearrow}_N(r)$ evaluates to
\begin{subequations}\label{eq:Frfull}
\begin{align}\label{eq:frodd}
f^{\nwsearrow}_N(r)=\frac{1}{4}(N+r+1)^{2}+\frac{1}{2}(N-r-1),
\end{align}
For $N+r$ even, we obtain
\begin{align}\label{eq:freven}
f^{\nwsearrow}_N(r)=\frac{1}{4}(N+r)^{2}+\frac{1}{2} (N-r)+2,
\end{align}
except for the special cases $N+r=10$ and $N+r=16$, where we obtain
\begin{align}\label{eq:frspecial}
f^{\nwsearrow}_N(r)=\begin{cases}34-r, & \quad r=10-N\quad(N\geq 8)\\
76-r, & \quad r=16-N\quad(N\geq 12)
\end{cases}.
\end{align}
\end{subequations}
\end{lemma}
The function $f^{\nwsearrow}_N(r)$ is illustrated in Fig.~\ref{fig:Fmaxr} for $N=20$.
\begin{figure}[tb]
\centering
\includegraphics[width=0.49\textwidth]{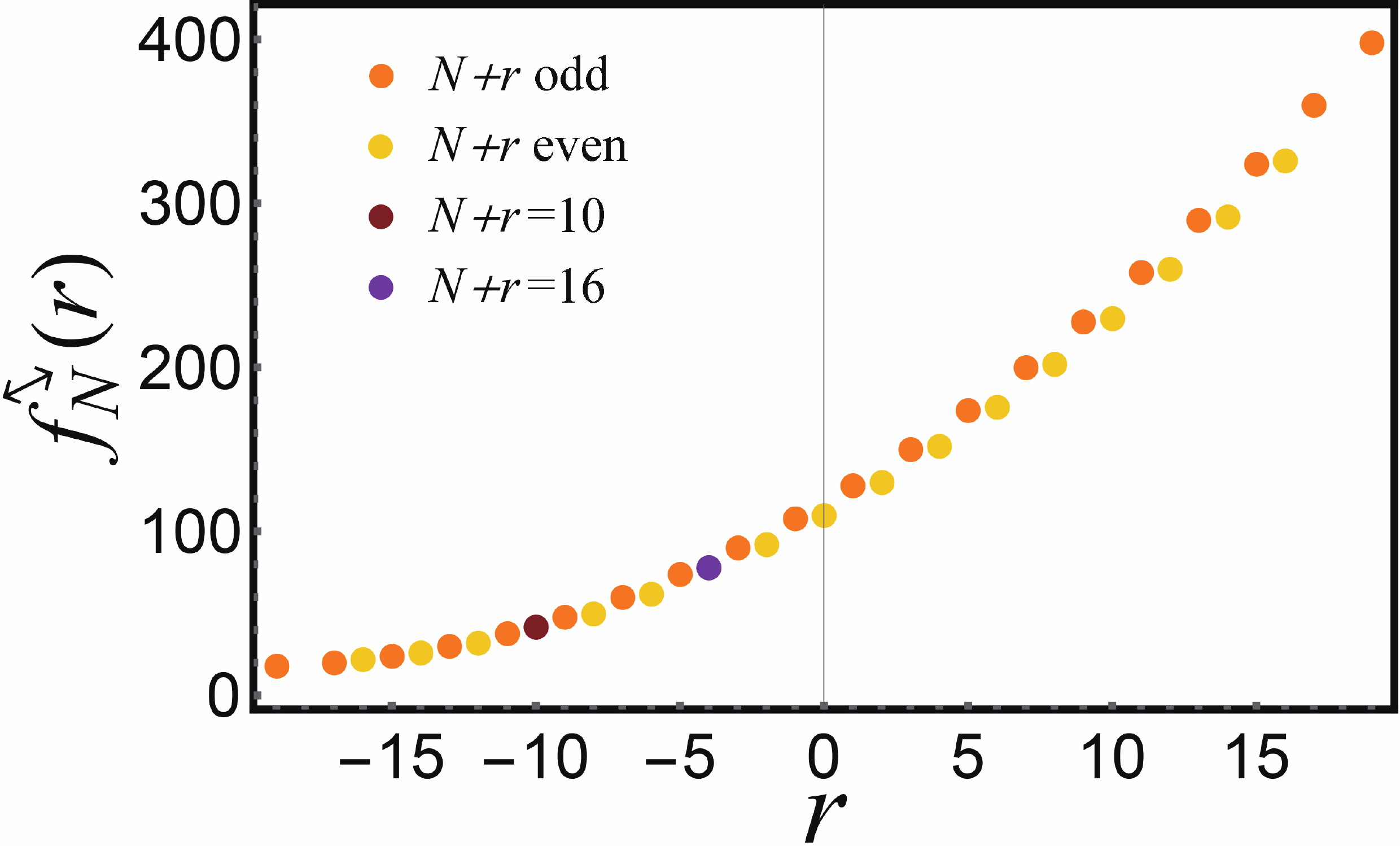}
\caption{
Plot of $f^{\nwsearrow}_N(r)$ for $N=20$. The orange dots present $f^{\nwsearrow}_N(r)$ for $N+r$ is odd as in Eq.~(%
\protect\ref{eq:frodd}). The yellow dots are for $N+r$ even as in Eq.~(%
\protect\ref{eq:freven}). The brown dot describes $r=10-N$ and the purple dot $r=16-N$ as per Eq.~(%
\protect\ref{eq:frspecial}).}
\label{fig:Fmaxr}
\end{figure}
\begin{proof} We can define $f^{\nwsearrow}_N(r)$ in two equivalent ways as $f^{\nwsearrow}_N(r)=\max_hf_N(r+h,h)=\max_wf_N(w,w-r)$. In the following, we focus on the latter and we identify the value of $w$ that maximizes the expression.\\
\textbf{Case 1}: $N+r$ is odd. For fixed $r$, the width $w$ is bounded from
above by $w\leq w_{\max }=\frac{N+r+1}{2}$. The function $f_N(w,r) $ is indeed maximized when $w=w_{\max}$. To see this, first note that inserting $w=w_{\max}$ into Eq.~(\ref{eq:fNwhfull}) yields $k=\lfloor \frac{N-w+r}{w-1}\rfloor =1$, and we obtain
\begin{align}
f_N(w_{\max},w_{\max}-r)&=\frac{1}{4}(N+r+1)^{2}+\frac{1}{2}(N-r-1) \notag\\
&=N+w_{\max}(w_{\max}-1).  \label{o FI wmax r}
\end{align}
In order to show that this is indeed the maximum possible value, we show that
any other choice for $w$ (which necessarily must be lower) at fixed $r$ will
lead to a smaller or equal value for $f_N(w_{\max},w_{\max}-r)$. Reducing the
value of $w$ to $w^{\prime }=w_{\max }-\delta $, where $1\leq \delta \leq
w_{\max }-1$ (the upper bound follows from the condition $w^{\prime }\geq
1$), we obtain
\begin{align}
f_{N}(w^{\prime },w'-r)=f_{N}(w_{\max },w_{\max }-r)+\Delta (w_{\max
},\delta ),  \label{o FI wmax-m r}
\end{align}%
where
\begin{equation*}
\Delta (w_{\max },\delta )=x^{2}+(w_{\max }-5\delta -1)x-\delta (2w_{\max
}-5\delta -3),
\end{equation*}%
and $x=(w_{\max }-\delta -1)\left\lfloor \frac{2\delta }{w_{\max }-\delta -1}%
\right\rfloor $.

Next, we determine the range of $x$. As $q-1<\lfloor q\rfloor\leq q$ for all $q$, it follows by taking $q=a/b$ that
\begin{align}
\forall a>0,\:\forall b\in\mathbb{R}:& & b-a&< a\left\lfloor \frac{b}{a}%
\right\rfloor\leq b  \label{eq:lemma1} \\
\forall a<0,\:\forall b\in\mathbb{R}:& & b&\leq b\left\lfloor \frac{b}{a}%
\right\rfloor< b-a .  \label{eq:lemma2}
\end{align}
Using $w_{\max }\geq \delta +1$, we obtain from Eq.~(\ref{eq:lemma1}) that $%
x_1< x\leq x_2$, where we introduced $x_1=3\delta -w_{\max }+1$ and $%
x_2=2\delta$.

We observe that $\Delta (w_{\max },\delta )$ is a parabola opening upwards with $%
\Delta (w_{\max },\delta )|_{x=x_i}=-\delta (\delta -1)\leq 0$ for $i=1,2$.
The function $\Delta (w_{\max },\delta )$ is thus negative for all values of
$x$ in its range between $x_{1}$ and $x_{2}$, i.e., $%
\Delta (w_{\max },\delta )\leq 0$. In Eq.~(\ref{o FI wmax-m r}) this
implies that $f_{N}(w_{\max},w_{\max}-r)\geq f_{N}(w_{\max}-\delta,w_{\max}-\delta-r)$ for all $\delta $, which proves the result: For odd $N+r$, we have
\begin{align}
f^{\nwsearrow}_N(r)=\max_{w}f_{N}(w,w-r)=\frac{1}{4}(N+r+1)^{2}+\frac{1}{2}(N-r-1),  \label{odd FI max r}
\end{align}
and the maximum is achieved by $w=\frac{N+r+1}{2}$.

\textbf{Case 2}: $N+r$ is even. For fixed $r$, the width $w$ is upper
bounded by $w\leq w_{\max}=\frac{N+r}{2}$. We show that in this case, the function $f_{N}(w,w-r)$ is maximized when $w=w_{\max}$ with the result
\begin{align}
f_{N}(w_{\max},w_{\max}-r)=\frac{1}{4}(N+r)^{2}+\frac{1}{2}(N-r)+2,
\label{e FI wmax r}
\end{align}
except when $N+r=10$ or $N+r=16$.

As $r$ is limited to the values
\begin{align}\label{eq:rrange}
r\in\{-(N-1),&-(N-1)+2,-(N-1)+3,\notag\\&\dots,(N-1)-3,(N-1)-2,N-1\},
\end{align}
the smallest even value of $N+r$ is $4$. In this case, using Eq.~(\ref{eq:fNwhfull}) with $w=w_{\max}$, we get $k=\lfloor\frac{N+r}{N+r-2}\rfloor=2$ and $f_N(w_{\max},w_{\max}-r)=N+4$. When $N+r>4$, we find $%
k=1 $ and $f_N(w_{\max},w_{\max}-r)=\frac{1}{4}(N+r)^{2}+\frac{1}{2}(N-r)+2$. In
short, we can express the result for arbitrary even values of $N+r$ as in Eq.~(\ref{e FI wmax r}).

To prove that this value is maximal in the stated cases, we proceed as in the odd case before. Reducing $w$ to $w^{\prime }=w_{\max }-\delta $, where $1\leq \delta
\leq w_{\max }-1$, we obtain $f_{N}(w^{\prime },r)=f_N(w_{\max },w_{\max}-r)+\Delta (w_{\max },\delta )$ with
\begin{align}
\Delta (w_{\max },\delta )=x^{2}+(w_{\max }-5\delta -3)x+\delta (5\delta
-2w_{\max }+7),  \label{e Delta F wmax w-m}
\end{align}%
and $x=(w_{\max }-\delta -1)\left\lfloor \frac{2\delta +1}{w_{\max
}-\delta -1}\right\rfloor $ is limited to the range $3\delta +2-w_{\max }< x\leq 2\delta
+1$.

We begin by demonstrating that $\Delta (w_{\max },\delta )\leq\Delta
(w_{\max }^{0},\delta )$ holds for all values of $w_{\max }$, where $w_{\max
}^{0}=3\delta +2$. At $w_{\max }=w_{\max }^{0}$, we have $x=2\delta +1$ and $%
\Delta (w_{\max }^{0},\delta )=\delta (3-\delta )$. For $%
w_{\max }>w_{\max }^{0}$, we obtain $x=0$ and
\begin{align}
\Delta (w_{\max },\delta )=\Delta _{0}(w_{\max },\delta ):=\delta (5\delta
-2w_{\max }+7).  \label{e Delta F wmax0}
\end{align}
For all $w_{\max }> w_{\max}^{0}$ we obtain $\Delta (w_{\max },\delta )<\Delta (w_{\max }^{0},\delta)$. For the opposite scenario, $%
w_{\max }<w_{\max }^{0}$, we introduce $d>0$ such that $w_{\max }=3\delta
+2-d$. Choosing $x$ at the extreme values of $x_{1}=3\delta +2-w_{\max }$ or $x_{2}=2\delta +1$ yields  $\Delta (w_{\max },\delta ) = \delta (3-\delta )-d$. As a function of $x$, the parabola $\Delta (w_{\max },\delta )$ opens upwards, and we thus have $\Delta (w_{\max },\delta )\leq \delta (3-\delta )-d<\delta (3-\delta )$ for all $x\in[x_1,x_2]$. Combining all cases, we conclude that
\begin{align}
\Delta (w_{\max },\delta )\leq \delta (3-\delta )
\label{e Delta F wmax bound}
\end{align}%
holds for all $w_{\max }$.

Since $\Delta (w_{\max },\delta )\leq
0 $ for $\delta \geq 3$ we can limit our attention to the cases $\delta=1,2$. For $\delta =1$,
we obtain $w_{\max }^{0}=5$ and the corresponding $\Delta
(5,1)=2$. For all $w_{\max}>w^0_{\max}$, we have $\Delta (w_{\max },1)=12-2w_{\max }\leq 0$. For the remaining possibilities $w_{\max}\in\{3,4\}$, we get $\Delta
(w_{\max },\delta )=0$. Hence, for $\delta=1$, we achieve $\Delta (w_{\max },1
)<0$ only at $w_{\max}=5$, i.e., when $N+r=10$. The partition that achieves the maximum has $k=2$ entangled sets of $w=4$ particles (Fig.~\ref{fig:3}) and this scenario is possible only for $N\geq 8$. In this case, we obtain with Eq.~(\ref{eq:fNwhfull}) that $f^{\nwsearrow}_N(r)=f_N(w_{\max}-1,w_{\max}-1-r)=34-r$.

Similarly, in the case of $\delta =2$ we obtain $w_{\max}^{0}=8$ and $\Delta (8,2)=-2$. Following a similar argument, this implies that $\Delta (w_{\max },2
)<0$ can only be achieved by $w_{\max}=8$, i.e., for $N+r=16$. The optimal partition has $k=2$ entangled sets of size $w=6$, and this is possible only when $N\geq 12$. In this case, Eq.~(\ref{eq:fNwhfull}) yields $f^{\nwsearrow}_N(r)=f_N(w_{\max}-2,w_{\max}-2-r)=76-r$.

To summarize, for $N+r$ even, we find that
\begin{align}
f^{\nwsearrow}_N(r)=\frac{1}{4}(N+r)^{2}+\frac{1}{2}%
(N-r)+2,  \label{even FI max r 1}
\end{align}
and this value is achieved by $w=w_{\max}=\frac{N+r}{2}$ except when $N+r=10$ and $N\geq 8$ or $N+r=16$ and $N\geq
12$. In these cases optimal partitions have $k=2$ and for $N+r=10$, the maximum is achieved by $f^{\nwsearrow}_N(r)=f_N(w_{\max}-1,w_{\max}-1-r)=34-r$, whereas for $N+r=16$ we find $f^{\nwsearrow}_N(r)=f_N(w_{\max}-2,w_{\max}-2-r)=76-r$.
\end{proof}

Again, we can find a simpler upper bound $ f^{\nwsearrow}_N(r)\leq  f^{\nwsearrow,\mathrm{up}}_N(r)$ for $N+r\neq 4$:
\begin{align}
 f^{\nwsearrow,\mathrm{up}}_N(r)=\left(\frac{N+r+1}{2}\right)\left(\frac{N+r-1}{2}\right)+N. \label{eq:Fr bound}
\end{align} 
leading to Eq.~(\ref{eq:Fr}) in the main text. 
\begin{proof}Note that this bound coincides with the expression~(\ref{eq:frodd}) for $N+r$ odd. In the even case, the difference between Eq.~(\ref{eq:Fr bound}) and Eq.~(\ref{eq:freven}) is 
\begin{align}
\delta= f^{\nwsearrow,\mathrm{up}}_N(r)-f^{\nwsearrow}_N(r)=\frac{N+r}{2}-\frac{9}{4}, \label{eq:Fr bound minus even}
\end{align}
with $\delta>0$ for $N+r>4$. For the smallest even value of $N+r=4$, we have $\delta<0$. The tight bound~(\ref{e FI wmax r}) yields $N+4$ for $N+r=4$. Finally, it is easily verified that the bound Eq.~(\ref{eq:Fr bound}) exceeds Eq.~(\ref{eq:frspecial}) in the given special cases.
\end{proof}

\begin{figure}[tb]
\centering
\includegraphics[width=0.49\textwidth]{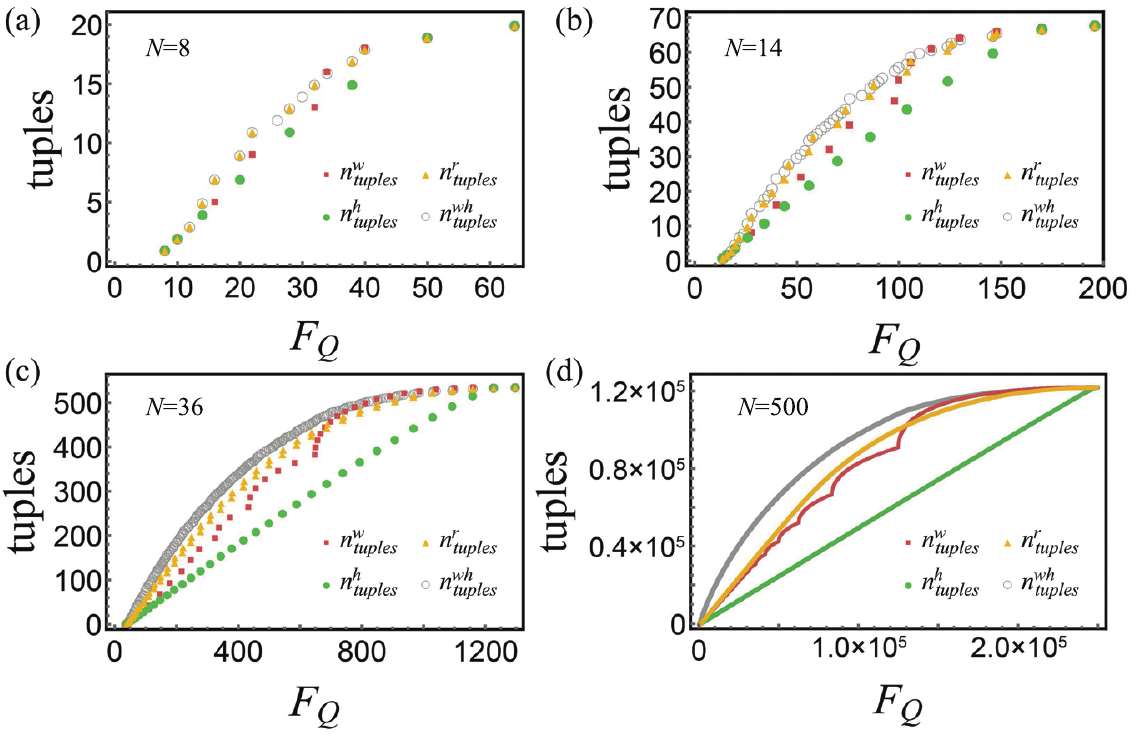}
\caption{Excluded tuples by comparing the Fisher information to separability bounds. We show the excluded tuples by $f^{\leftrightarrow}_N(w)$ (red dots), $f^{\updownarrow}_N(w)$ (green dots), $f^{\nwsearrow}_N(r)$ (yellow dots), and $f_N(w,h)$ (gray dots). Generally, $f^{\nwsearrow}_N(r)$ excludes more tuples than the alternative single integer quantifiers based on $w$ and $h$ when $F_Q$ is small. The entanglement depth, i.e., $f^{\leftrightarrow}_N(w)$ provides most information when $F_Q$ is close to its maximal value $N^2$.}
\label{fig:01}
\end{figure}

\section{Number of excluded tuples by different bounds}
Any measured lower bound on $F_Q$ now imposes limits on the possible values of $w$, $h$, $r$ and $(w,h)$ via the bounds~(\ref{eq:FQmaxw}), (\ref{eq:FQmaxh}), (\ref{eq:FQmaxwh}), and~(\ref{eq:FQmaxr}), respectively. Each of these bounds excludes a different amount of tuples $(w,h)$. To gauge the power of the different entanglement witnesses, we quantify the number of tuples $(w,h)$ that are identified as inseparable as a function of the value of $F_Q[\hat{\rho},\hat{J}_{\mathbf{n}}]$.
For tuples with fixed $w$, their height $h$ has a domain of $\lceil N/w\rceil\leq h\leq N+1-w$%
, i.e., there are $N+2-w-\left\lceil \frac{N}{w}\right\rceil$ tuples with
fixed $w$. The total number of tuples with width no larger than $w$ is then given by the sum
\begin{align}
n^w_{\mathrm{tuples}}(w)&=\sum_{w_{i}=1}^{w}\left(N+2-w_{i}-\left\lceil \frac{N}{w_{i}}\right\rceil\right)\notag\\
&=\frac{w}{2}(2N-w+3)-\sum_{w_{i}=1}^{w}\left\lceil \frac{N}{w_{i}}\right\rceil.
\end{align}
Analogously, making use of the bounds $\lceil N/h\rceil\leq w\leq N+1-h$ we obtain that there exist $N+2-h-\lceil N/h\rceil$ tuples with fixed $h$, leading to a total number of
\begin{align}
n^h_{\mathrm{tuples}}(h)&=\sum_{h_{i}=h}^{N}\left(N+2-h_{i}-\left\lceil \frac{N}{h_{i}}\right\rceil\right)\notag\\
&=\frac{1}{2}(N-h+4)(N-h+1)-\sum_{h_{i}=h}^{N}\left\lceil \frac{N}{h_{i}}\right\rceil
\end{align}
tuples with height larger or equal to $h$. Finally, using $\left\lceil\frac{\sqrt{r^{2}+4N}+r}{2}\right\rceil\leq w\leq \left\lfloor\frac{N+1+r}{2}\right\rfloor$, and with the possible values of $r$ given in Eq.~(\ref{eq:rrange}), we obtain $n^r_{\mathrm{tuples}}(1-N)=1$, $n^r_{\mathrm{tuples}}(N-1)=N(N+3)/2+\sum_{w=1}^N\left\lceil\frac{N}{w}\right\rceil$ and for the remaining values of $3-N\leq r\leq N-3$, we obtain
\begin{align}
&\quad n^r_{\mathrm{tuples}}(r)\\&=N+r-1+\sum_{r_{i}=3-N}^{r}\left(\left\lfloor \frac{%
N+1+r_{i}}{2}\right\rfloor -\left\lceil \frac{\sqrt{r_{i}^{2}+4N}+r_{i}}{2}%
\right\rceil \right)\notag
\end{align}
tuples with rank no larger than $r$. The number $n^{wh}_{\mathrm{tuples}}(w,h)$ of tuples excluded by comparing to the bounds on $(w,h)$ depends on the value of $F_Q$ and is determined numerically. The number of excluded tuples is compared in Fig.~\ref{fig:01}.

\section{Separability bounds on the spin-squeezing coefficient}\label{sec:xi}
The Wineland \textit{et al.} spin-squeezing coefficient fulfils the following upper bound for arbitrary $\Lambda$-separable states~\cite{FadelPRA2020}
\begin{align}\label{xi-2bound}
N\xi _{\rho_{\Lambda} }^{-2}\leq \frac{1}{2}\sum_{l=1}^{|\Lambda|}N_{l}^{2}+N.
\end{align}
This bound can be asymptotically saturated when all $N_l>1$. Inserting~(\ref{eq:maxLambda}), which can always be saturated by an optimal state, and maximizing over all states yields the expression~(\ref{eq:xiFQmax}) in the main text. Hence, we find that for all $(w,h)$-separable states
\begin{align}
N(\xi_{\rho_{(w,h)-\mathrm{sep}}}^{-2}-1)\leq \frac{1}{2} f_N(w,h).
\end{align}
The left-hand side is the difference between the sensitivity and the shot-noise limit. Analogous results hold for states with finite $w$, $h$, or $r$, using the corresponding functions $f^{\leftrightarrow}_N(w)$, $f^{\updownarrow}_N(h)$, and $f^{\nwsearrow}_N(w)$, respectively.

We find the separability limit for the spin-squeezing coefficient, e.g., for $(w,h)$-separable states
\begin{align}
\xi_{\rho_{(w,h)-\mathrm{sep}}}^{2}&\geq \frac{2N}{f_N(w,h)+2N},
\end{align}
and the explicit expression for $f_N(w,h)$ can be found in Eq.~(\ref{eq:fNwhfull}). The resulting bound is plotted in Fig.~\ref{fig:1}~(b) in the main text on dB scale. We find analogous expressions for the bounds on $w$, $h$ and $r$ from the respective bounds~(\ref{eq:flrw}), (\ref{eq:fudh}), and~(\ref{eq:Frfull}). Using the upper bounds that ignore the division into integer subsets, we obtain simpler expressions. For instance, Eq.~(\ref{eq:fNwhupper}) implies that $(w,h)$-separable states cannot reduce the phase uncertainty below
\begin{align}\label{eq:xiwh}
\xi_{\rho_{(w,h)-\mathrm{sep}}}^{2}&\geq \frac{2N}{w(N-h)+3N}. 
\end{align}
By projecting only on the information provided by $w$ we recover the result~\cite{FadelPRA2020}
\begin{align}\label{eq:xiw}
\xi_{\rho_{w-\mathrm{prod}}}^{2}&\geq \frac{1}{1+w/2},
\end{align}
and similarly for $h$:
\begin{align}\label{eq:xih}
\xi_{\rho_{h-\mathrm{sep}}}^{2}&\geq \frac{2N}{(N-h+1)^2+h-1+2N}.
\end{align}
Finally, the sensitivity gain is bounded in terms of $r$ as
\begin{align}\label{eq:xirSupp}
\xi_{\rho_{r-\mathrm{rnk}}}^{2}&\geq \frac{2N}{\left(\frac{N+r+1}{2}\right)\left(\frac{N+r-1}{2}\right)+3N},
\end{align}
which coincides with Eq.~(\ref{eq:xir}) in the main text.

\section{Analysis of experimental data}
To analyze experimental data, we compare measured lower bounds for $F_Q$ or $\xi^2$ to the bounds~(\ref{eq:FQmaxwh}), (\ref{eq:FQmaxw}), (\ref{eq:FQmaxh}), and~(\ref{eq:FQmaxr}). In Fig.~\ref{fig:4} we show how $F_Q$ interpolates between the separable bound $N$ and the genuine multipartite entangled quantum limit $N^2$ as a function of $w$, $h$, and $r$ for small values of $N$.
\begin{figure}[tbp]
\centering
\includegraphics[width=0.35\textwidth]{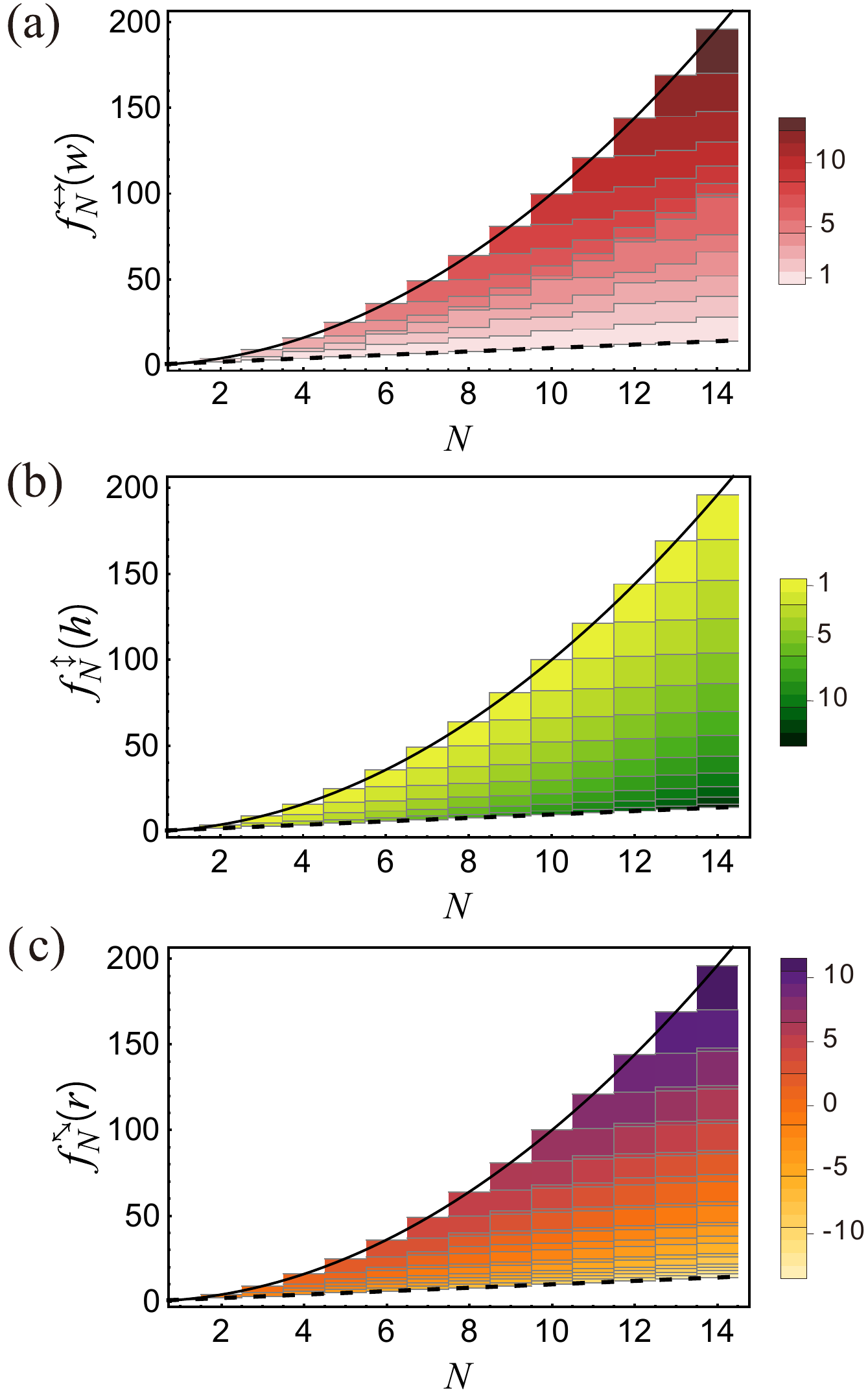}
\caption{Maximal values of the quantum Fisher information as a function of $w$ (a)~\cite{PezzePNAS2016}, $h$ (b), and $r$ (c) for $N$ up to $14$.}
\label{fig:4}
\end{figure}
In the main text, the data on $F_Q$ for $N=14$ ions from Ref.~\cite{MonzPRL2011} was analyzed in detail in Fig.~3, while summaries were presented for the data on $F_Q$ with $N=127$ and $\xi^2$  from~\cite{BohnetSCIENCE2016} and with $N=470$ from~\cite{StrobelSCIENCE2014} in Fig.~4. Here we present the full analysis of the data from~\cite{BohnetSCIENCE2016,StrobelSCIENCE2014}, and we analyze additional experimental data from Refs.~\cite{MonzPRL2011,BarontiniSCIENCE2015}.

\subsection{Data on $F_Q$}
In the main text the data for $N=14$ ions from~\cite{MonzPRL2011} was analyzed. In the same reference, the state closest to genuine multipartite entanglement was produced with $N=8$ ions with $F_Q\geq 39.6$. The corresponding analysis is shown in Fig.~\ref{fig:7}. The width $w$ obtained from Eq.~(\ref{eq:FQmaxw}) is $6$ and the excluded tuples $(w,h)$ are $16$. The height $h$ obtained from Eq.~(\ref{eq:FQmaxh}) is $2$ and the excluded tuples $(w,h)$ are $15$. Dyson's rank $r$ obtained from Eq.~(\ref{eq:FQmaxr}) is $4$ and the excluded tuples $(w,h)$ are $17$. From~(\ref{eq:FQmaxwh}), we can exclude $17$ tuples $(w,h)$ by using full information on the tuple $(w,h)$.
\begin{figure}[tbp]
\centering
\includegraphics[width=0.49\textwidth]{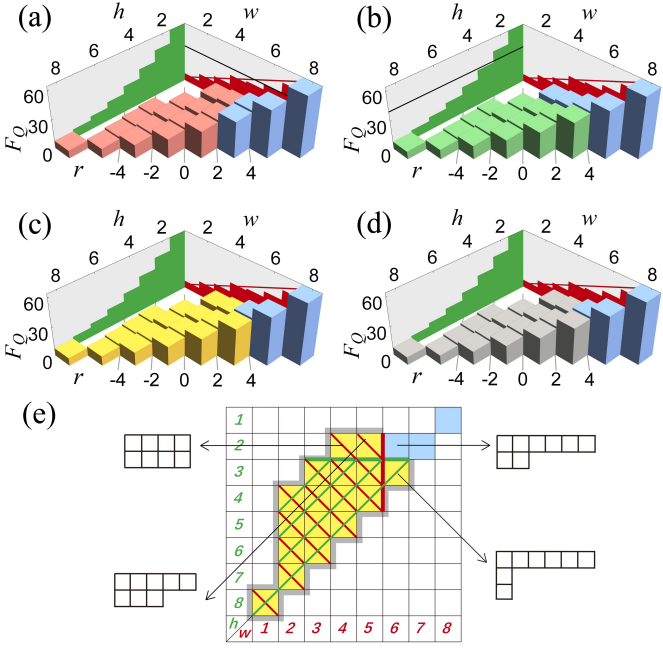}
\caption{The entanglement analysis (same as Fig.~\ref{fig:2} in the main text) of $N=8$ ions reported in~\cite{MonzPRL2011} whose measurements allow to infer a minimal value of $F_Q\geq 39.6$~\cite{PezzePNAS2016}. Since the measured value of $F_Q$ is closer to its maximal value of $N^2=64$, the entanglement depth $w$ provides a good characterization of the entangled tuples.}
\label{fig:7}
\end{figure}

In Fig.~\ref{fig:127} we show the full analysis of the measurement of $F_Q\geq 266.7$ from Ref.~\cite{BohnetSCIENCE2016} for $N=127$ trapped ions. We can exclude $64$ tuples $(w,h)$ from $w=3$, obtained from Eq.~(\ref{eq:xiw}) and $67$ tuples are excluded from $h=115$ via Eq.~(\ref{eq:xih}). Dyson's rank $r=-102$ obtained from Eq.~(\ref{eq:xirSupp}) excludes $133$ tuples and by using the full information provided by~(\ref{eq:xiwh}), we can exclude $236$ tuples.

In Fig.~\ref{fig:5}, we analyze the experimental data from the ultracold atom experiment of Ref.~\cite{BarontiniSCIENCE2015}. The reported data shows $F_Q\geq 54.36$ with $N=36$ atoms. The entanglement depth obtained from Eq.~(\ref{eq:FQmaxw}) is $w=2$, which only excludes the fully separable partition. The separability $h$ obtained from Eq.~(\ref{eq:FQmaxh}) is $32$ and the excluded tuples $(w,h)$ are $7$. Dyson's rank $r$ obtained from Eq.~(\ref{eq:FQmaxr}) is $-27$, which is able to exclude $13$ tuples. We can exclude $18$ tuples $(w,h)$ by using the full information on $(w,h)$, i.e., the bound~(\ref{eq:FQmaxwh}).

\begin{figure}[tbp]
\centering
\includegraphics[width=0.49\textwidth]{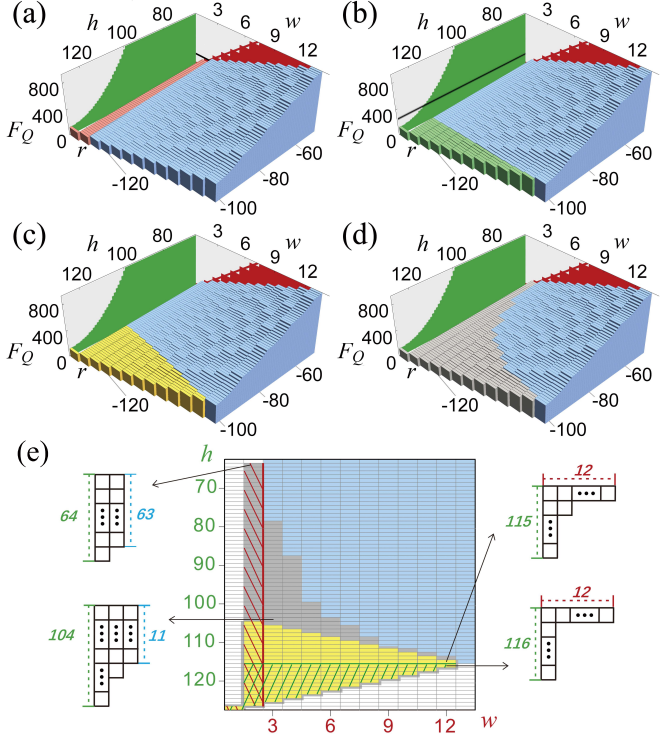}
\caption{Full entanglement analysis (same as Fig.~\ref{fig:2} in the main text) of $N=127$ ions reported in~\cite{BohnetSCIENCE2016} with $F_Q\geq 266.7$.}
\label{fig:127}
\end{figure}

\begin{figure}[tbph]
\centering
\includegraphics[width=0.49\textwidth]{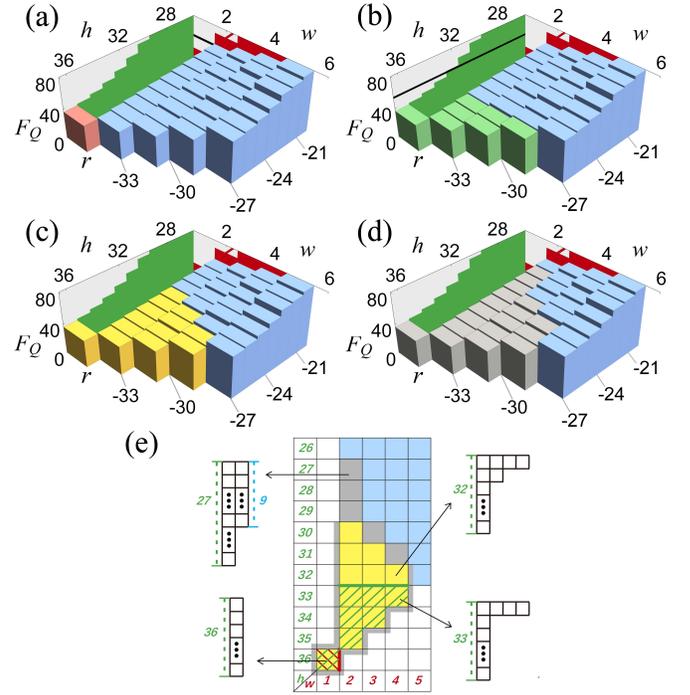}
\caption{The entanglement analysis (same as Fig.~\ref{fig:2} in the main text) of $N=36$ ultracold atoms from \cite{BarontiniSCIENCE2015} with $F_Q\geq 54.36$.}
\label{fig:5}
\end{figure}

\begin{figure}[tbph]
\centering
\includegraphics[width=0.49\textwidth]{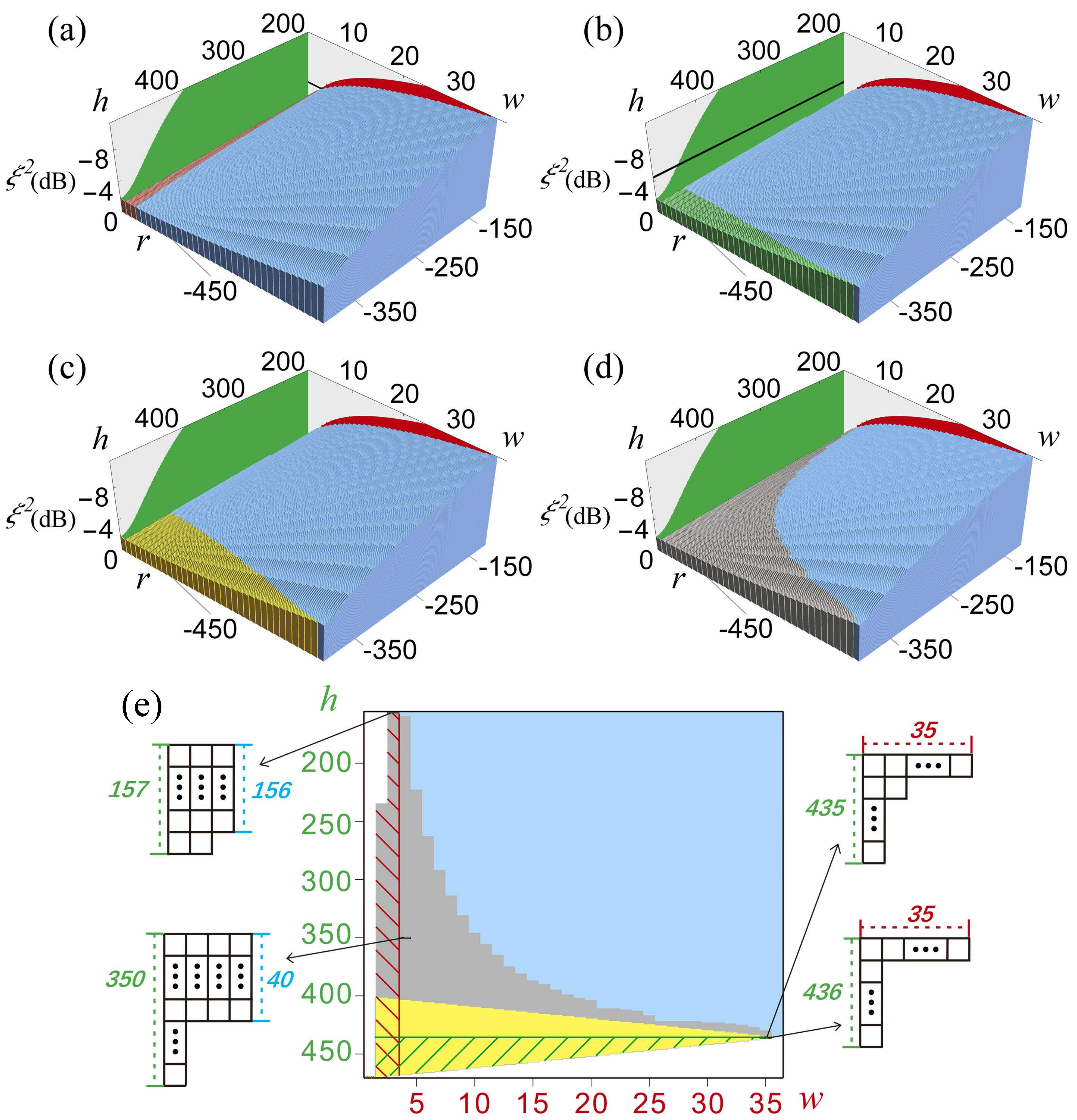}
\caption{The entanglement analysis, based on a measurement of the spin-squeezing parameter of $\xi^2\leq -4.5\,\mathrm{dB}$ with $N=470$ cold atoms, reported in Ref.~\cite{StrobelSCIENCE2014}. The characterization of multipartite entanglement in terms of $w$, $h$, $r$ and $(w,h)$ is analogous to that presented previously (see, e.g., Fig.~\ref{fig:2} in the main text). Here, however, we make use of the corresponding bounds on the spin-squeezing coefficient that were derived in Sec.~\ref{sec:xi}.}
\label{fig:470}
\end{figure}
In Fig.~\ref{fig:470}, we further show the full analysis of the data provided in Fig.~4~(b) of the main text.

\subsection{Data on $\xi^{-2}$}
Measurements of $\xi^2$ are experimentally less demanding than those of the quantum Fisher information $F_Q$. In the literature a large amount of results on $\xi^2$ have been published. Here, we pick the experimental data from~\cite{StrobelSCIENCE2014} as an example, where $N=470$ atoms were prepared with $\xi^2\leq-4.5\,\mathrm{dB}$ of squeezing. The full analysis is shown in Fig.~\ref{fig:470}. We can exclude $548$ tuples $(w,h)$ from $w=4$, obtained from Eq.~(\ref{eq:xiw}) and $596$ tuples are excluded from $h=435$ via Eq.~(\ref{eq:xih}). To compare with experimental data, we make use of the tightest versions of these bounds, based on the fully general expression~(\ref{eq:fNwhfull}). Dyson's rank $r=-399$ obtained from Eq.~(\ref{eq:xirSupp}) excludes $1191$ tuples and by using the full information provided by~(\ref{eq:xiwh}), we can exclude $2941$ tuples.


\begin{thebibliography}{99}
\bibitem{footnote} In the plots and analysis of experimental data, we use slightly stronger versions of the bounds~(\ref{eq:Fwh}), (\ref{eq:Fw}), (\ref{eq:Fr}), and~(\ref{eq:xir}) that are obtained by considering the division of $N$ into integer subsets. The exact expressions and their derivation are provided in~\cite{Supp}. The red line in the projection of $w$ represents the bound~(\ref{eq:Fw}).
\end{thebibliography}
\end{document}